\documentclass[12pt]{article}

\usepackage{setspace}
\usepackage{amsmath}
\usepackage{amsthm}
\usepackage{hyperref, wasysym}
\usepackage{latexsym}
\usepackage{cite}
\usepackage{amssymb}
\usepackage{mathtools}
\usepackage{subcaption}
\newtheorem{theorem}{Theorem}

\newtheorem{corollary}[theorem]{Corollary}

\usepackage{lineno}
\usepackage[rmargin=0.75in,lmargin=0.75in,bmargin=0.735in,tmargin=0.5in]{geometry}
\usepackage{wasysym}
\usepackage{fancyhdr}
\usepackage{graphicx}
\usepackage{times}
\usepackage{color}
\usepackage{multicol}
\usepackage{placeins}
\usepackage{authblk}
\usepackage{arydshln}
\usepackage{float}
\usepackage{dsfont}
\usepackage{caption}
\usepackage{subcaption}
\usepackage{booktabs}
\usepackage{footmisc}
\usepackage{enumerate}
\usepackage{enumitem}
\usepackage{adjustbox}
\restylefloat{table}
\usepackage{lipsum}
\usepackage[labelfont=bf]{caption}

\def\>{\textgreater}

\begin{document}
\title{Noise-induced tipping under periodic forcing: preferred tipping phase in a non-adiabatic forcing regime}

\author[ ]{Yuxin Chen\footnote{Department of Engineering Sciences and Applied Mathematics, Northwestern University
    ({yuxinchen2018@u.northwestern.edu})} , \hspace{0.1cm}   John A. Gemmer\footnote{Department of Mathematics and Statistics, Wake Forest University ({gemmerj@wfu.edu})} , \hspace{0.1cm}  Mary Silber\footnote{Committee on Computational and Applied Mathematics and Department of Statistics, University of Chicago ({msilber@uchicago.edu})} , \hspace{0.05cm} and \hspace{0.05cm} Alexandria Volkening\footnote{Mathematical Biosciences Institute, Ohio State University ({volkening.2@mbi.osu.edu}) \\ \indent \textbf{Funding}: MS, YC, and AV are partially funded by the AMS Mathematics Research Communities (NSF Grant No. 1321794). The work of MS and YC is funded in part by NSF DMS-1517416. The work of AV and JG was supported in part by NSF-RTG grant DMS-1148284, and AV is currently supported  by the Mathematical Biosciences Institute and the NSF under DMS-1440386.}}
\maketitle

\begin{abstract}
We consider a periodically-forced 1-D Langevin equation that possesses two stable periodic solutions in the absence of noise. We ask the question: is there a most likely noise-induced transition path between these periodic solutions that allows us to identify a preferred phase of the forcing when tipping occurs? The quasistatic regime, where the forcing period is long compared to the adiabatic relaxation time, has been well studied; our work instead explores the case when these timescales are comparable. We compute optimal paths using the path integral method incorporating the Onsager-Machlup functional and validate results with Monte Carlo simulations. Results for the preferred tipping phase are compared with the deterministic aspects of the problem. We identify parameter regimes where nullclines, associated with the deterministic problem in a 2-D extended phase space, form passageways through which the optimal paths transit. As the nullclines are independent of the relaxation time and the noise strength, this leads to a robust deterministic predictor of preferred tipping phase in a regime where forcing is neither too fast, nor too slow.
\end{abstract}

\begin{paragraph}{Key words}
  noise-induced tipping, metastable state, optimal path, path integral method, stochastic dynamics
\end{paragraph}

\begin{paragraph}{AMS subject classifications}
68Q25, 60G17, 37J50
\end{paragraph}

\graphicspath{{./figs/}}
Mathematical mechanisms for `tipping points' in low-dimensional dynamical systems have been classified according to whether they involve, pre- dominantly, a bifurcation, noise, or parameter drift\cite{ashwin2012tipping,kuehn2011mathematical}. This paper focuses on noise-induced tip- ping, between distinct periodic attractors, in systems with periodic forcing. We are interested in determining whether there is a dominant phase of the forcing when the system is most likely to tip from one attractor to another, and, if so, which deterministic features might predict that phase. A better understanding of these features, and the possible role of intrinsic relaxation vs. exogenous forcing timescales, may help identify a phase of the forcing when the system is most vulnerable to a noise-induced abrupt transition.

In this paper, we focus on a simple stochastic differential equation for which we can readily control the characteristic timescales of the problem. Specifically, we consider
\begin{align} 
dX_t&=\frac{1}{\varepsilon} f(X_t,t)dt+\sigma dW_t\nonumber\\ &\equiv\frac{1}{\varepsilon}\left(X_t-X_t^3+\alpha +A\cos \left(2\pi t\right)\right)dt+\sigma dW_t.\label{eqn:ModelSDE}
\end{align}
Here $X_t\in \mathbb{R}$ is a stochastic process parameterized by time $t\geq 0$, $W_t$ is a standard Wiener process and $\sigma>0$ denotes the noise strength. We choose the parameters $\varepsilon, \alpha, A>0$ so that the deterministic problem has two stable periodic solutions separated by an unstable one (represented by solid and dashed curves in Figure \ref{Fig:Example} and \ref{fig: histogramsVarySigma}). The parameter $\varepsilon$ represents a ratio of the characteristic relaxation time of the flow to the period of the forcing, and arises when we nondimensionalize time by the forcing period. Throughout the paper, the ``lower'' and ``upper'' stable periodic solutions of $\dot{x}=f(x,t)/\varepsilon$ are denoted by $x_l^*$ and $x_u^*$, respectively, and the ``middle'' unstable one is denoted by $x_m^*$, where $x_l^*(t)<x_m^*(t)<x_u^*(t)$ for $t\in \mathbb{R}$. The presence of noise introduces a mechanism for transition from one stable periodic solution to the other. Figure \ref{Fig:Example} shows some example realizations of \eqref{eqn:ModelSDE}.

\begin{figure}[t]
\centering
\includegraphics[width=0.5\textwidth]{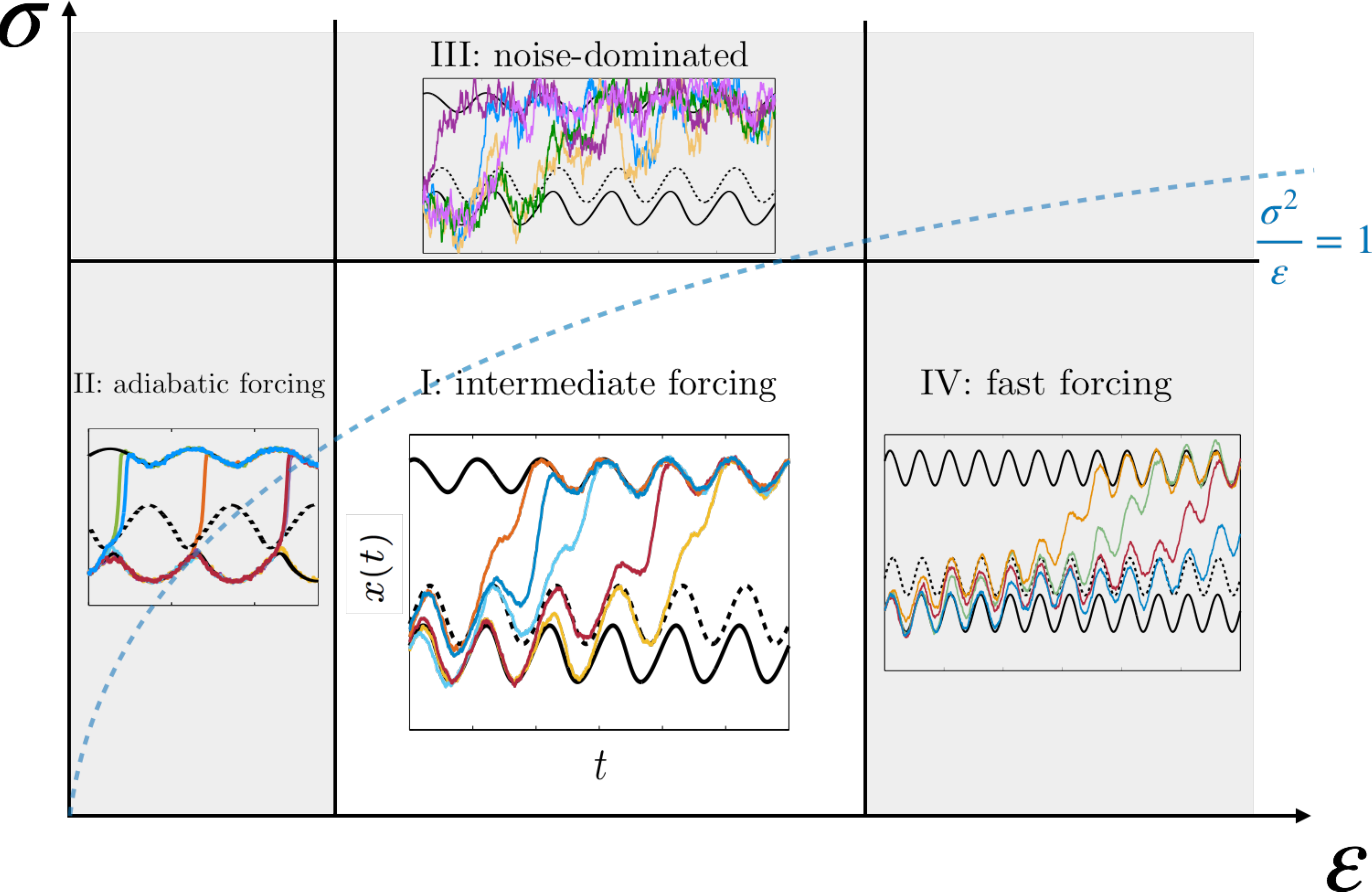}
\caption{Schematic parameter diagram in $(\varepsilon, \sigma)$-parameter plane. Insets are tipping realizations of Eq.\eqref{eqn:ModelSDE}. Black solid(dashed) curves represent stable(unstable) solutions to the deterministic Eq. \eqref{eqn:ModelSDE} with $\sigma = 0$.
Parameters in each inset are taken as follows: $(\varepsilon, \sigma) = (0.25, 0.1)$ in regime I; $(\varepsilon, \sigma) = (0.025, 0.1)$ in regime II; $(\varepsilon, \sigma) = (0.25, 0.8)$ in regime III; $(\varepsilon, \sigma) = (1, 0.2)$ in regime IV. $(\alpha, A) = (0.25, 0.5)$ in all regimes. }
\label{Fig:Example}
\end{figure}

In Figure \ref{Fig:Example} we present a schematic diagram of different $\varepsilon$ and $\sigma$ regimes in which noise-induced transitions differ qualitatively. Comparing insets in Figure \ref{Fig:Example}, the value of $\varepsilon$ in (\ref{eqn:ModelSDE}) can strongly impact the geometry and separation of the periodic solutions. This timescale ratio, in turn, may affect when and how the sample paths cross the unstable periodic solution. The adiabatic regime (regime II) in Figure \ref{Fig:Example}, where $\varepsilon\ll 1$, has been well-studied \cite{berglund2002metastability, berglund2003geometric, kuehn2011mathematical, kuehn2013mathematical,gammaitoni1998stochastic, freidlin2000quasi, berglund2002sample, herrmann2005exit}. In this case, the relaxation time to stable periodic orbits is fast compared to the forcing period. Consequently, noise induced transitions are ``instantaneous'' and occur with high probability when the separation between the stable and unstable periodic orbits is minimal and when an associated potential barrier height turns out to be smallest. 
In contrast, in regime I of Figure \ref{Fig:Example}, obtained with $\varepsilon=\mathcal{O}(1)$ and otherwise the same parameters, the transition occurs more slowly, especially after crossing the unstable periodic solution, and the separation between the deterministic solutions does not vary significantly over the forcing period. These observations suggest challenges with identifying a preferred tipping phase outside of the adiabatic regime. Moreover, in this regime noise is not too large, therefore noise and deterministic drift may act on comparable time-scales. 
In the noise-dominated regime (regime III in Figure \ref{Fig:Example}), randomness overpowers the deterministic dynamics and tipping samples are extremely noisy. In the fast forcing regime, marked as regime IV in Figure \ref{Fig:Example}, where the timescale ratio $\varepsilon$ is large, the relaxation time of the dynamics is slow compared to the forcing period and transitions can take a long time compared to the forcing period, possibly making the notion of preferred tipping phase mute.


In this paper, we study noise-induced transitions between stable periodic solutions in the intermediate forcing regime (regime I shown in Figure \ref{Fig:Example}). We use the framework of `exit problems' \cite{freidlin2012random}, where the most probable path of escape, from one domain of attraction to another,  is typically determined by taking the limit $\sigma\rightarrow 0$. In this setting, \emph{most probable paths} (also called \emph{optimal paths}) are often approximated by calculating the critical points of either the Freidlin-Wentzell (FW) \cite{freidlin2012random,day1990large, ren2004minimum} or the Onsager-Machlup (OM) rate functional \cite{onsager1953fluctuations,chaichian2001path}. Computing minimizers of these functionals can be an efficient means to estimate the optimal path without performing Monte Carlo simulations; for example, the FW rate functional was utilized to obtain optimal paths \cite{weinan2002string, heymann2008geometric, cameron2012finding} and the OM functional was applied for the same purpose \cite{ritchie2016early, navarra2013path}. While the FW and OM functionals are formally equivalent when $\sigma \to 0$,  their minimizers may differ for small but finite $\sigma$. Moreover, consideration of other small parameters in the problem, such as the time scale ratio $\varepsilon$, may be important in resolving discrepancies. 

As shown in Figure \ref{fig: histogramsVarySigma}, we let $t_l$, $t_m$, and $t_u$ denote the stopping times when a sample path departs from $x_l^*$, leaves $x_m^*$, and arrives at $x_u^*$, respectively. Formally, they are defined as $t_i = \min{\{ t\geq 0: X_t \geq x_i^*(t)\}}$ where $i = l, m$, and $u$. 
In Figure \ref{fig: histogramsVarySigma} we present the distributions of $t_l$, $t_m$, and $t_u$ modded by the forcing period $T=1$. We take deterministic parameters in regime I with $(\varepsilon, A, \alpha) = (0.25, 0.65, 0.15)$ and take the noise intensity with $\sigma = 0.15$ on the left panel and $0.3$ on the right panel. We notice that for both values of $\sigma$, the histograms of $t_l$ have a more prominent peak than the histograms of $t_m$ and $t_u$.  We find that $t_l$ is a robust `tipping phase' in this intermediate regime. Additionally, we find that the distribution of $t_l$ becomes less peaked as we increase the noise intensity $\sigma$, which indicates that there is no distinctive peak when $\sigma$ is too large and a most likely tipping phase may not be identifiable. Thus we do not consider the noise-dominated regime III in Figure \ref{Fig:Example}, where realizations are extremely noisy and deterministic features are washed out.  Histograms with different values of $\sigma \in [0.15, 0.4]$ (not shown here) for parameters in regime I give consistent results.

\begin{figure}[ht]
\centering
\includegraphics[width=0.4\textwidth]{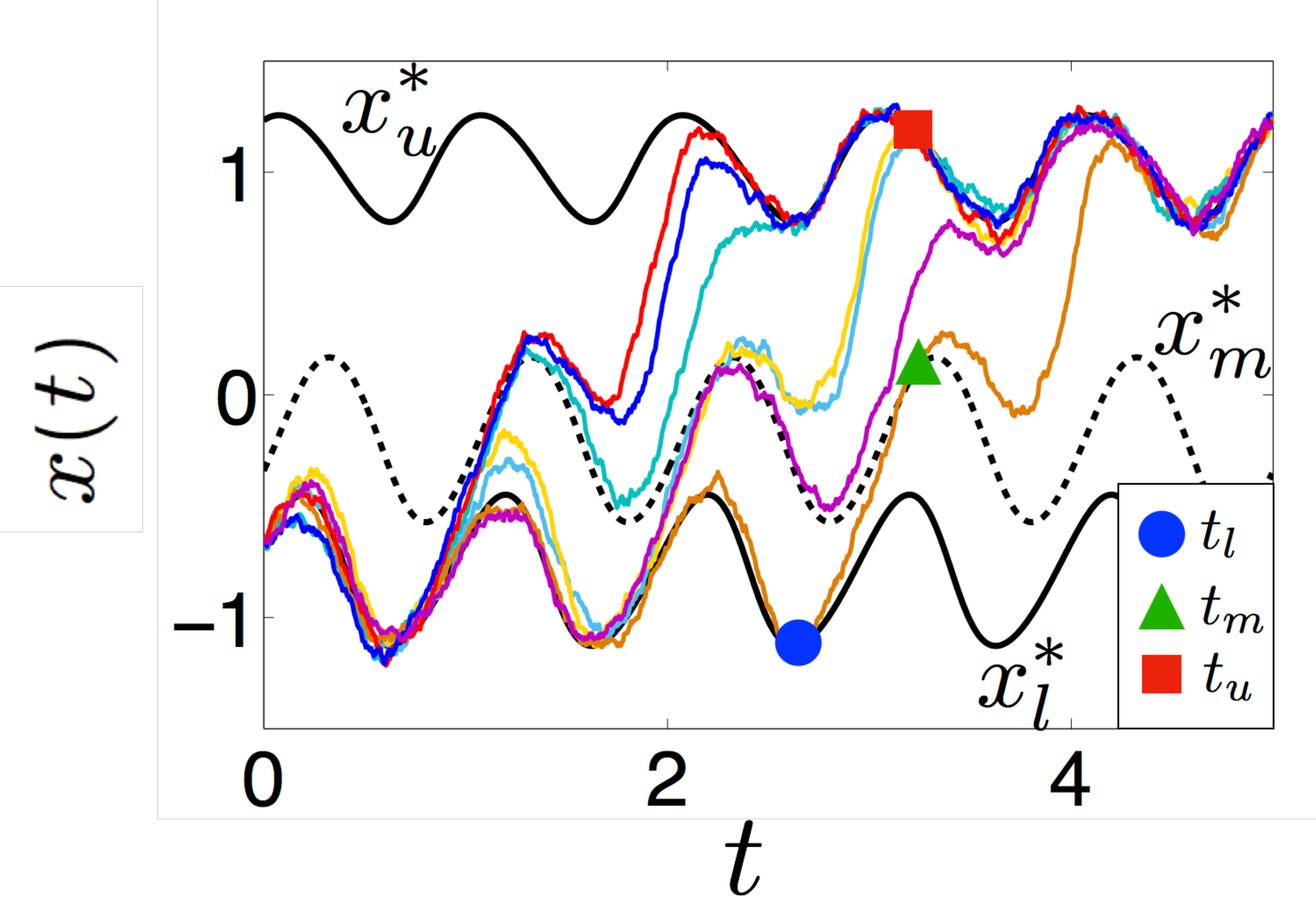}
\includegraphics[width=0.5\textwidth]{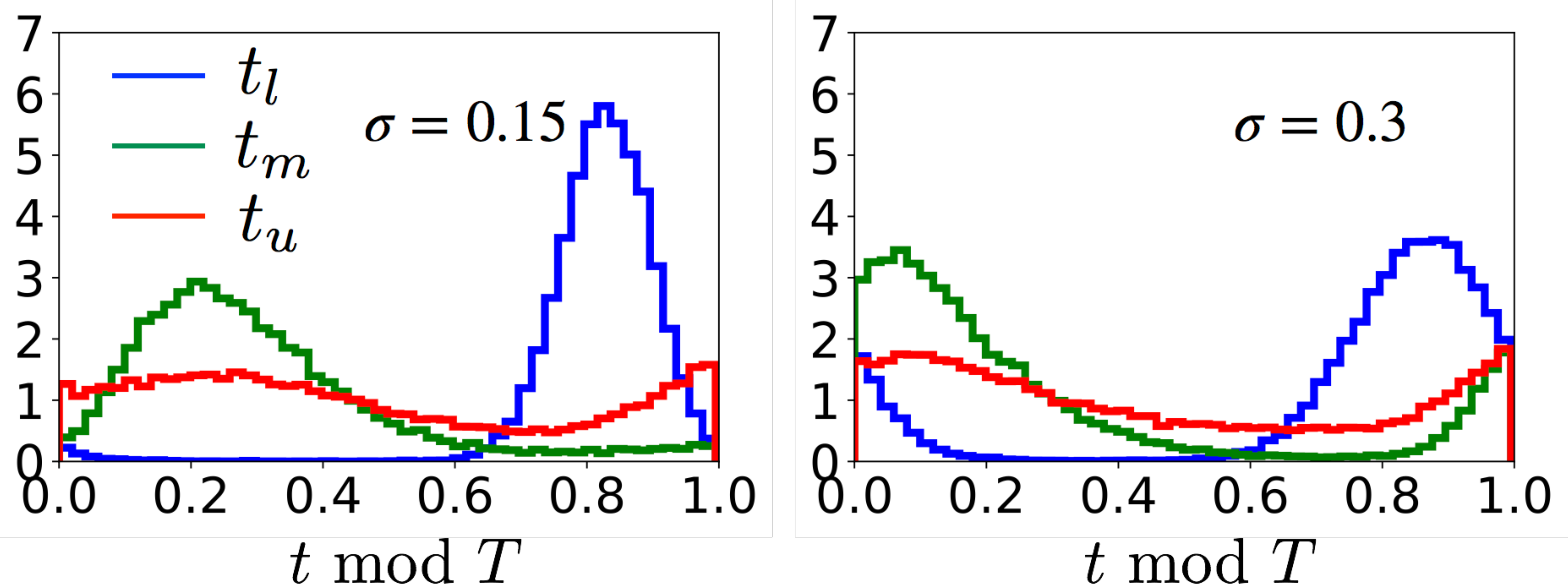}
\caption{Normalized histograms of $t_l, t_m$ and $t_u$, which denote the times when sample path of \eqref{eqn:ModelSDE} cross $x_l^*$, $x_m^*$, and $x_u^*$, respectively, with $(\varepsilon, A, \alpha) = (0.25, 0.65, 0.15)$. The noise intensity $\sigma$ is taken as $0.15$ on the left and $0.3$ on the right.}
\label{fig: histogramsVarySigma}
\end{figure}

In addition to looking for conditions on $\varepsilon$ and $\sigma$ under which we can identify a robust `tipping phase', we aim to investigate what features of the deterministic dynamics select such a phase when it exists. The rest of the paper is organized as follows. In Section \ref{sec:mathForm} we formulate the problem using  a path integral approach. We describe the FW and  OM functionals that we consider, and how the parameters $\varepsilon$ and $\sigma$ enter the variational problem for computing optimal paths. Section \ref{sec:noiseDriftBalance} contains the main contributions of this paper. Focusing on regime I where $\varepsilon = O(1)$, we perform numerical simulations over a range of parameters in order to identify a robust definition of preferred tipping phase. We find that it is associated with the time of initial departure from the stable periodic solution, as suggested by Figure \ref{fig: histogramsVarySigma}. It occurs at a phase of the forcing when the flow near the stable periodic solution changes, momentarily, from contracting to expanding. This phase is determined by an intersection of the periodic solution with the nullclines $f(x,t)=0$ of the deterministic system.  We conclude in Section \ref{sec:Discussion} with a discussion of our results, limitations of the method, and some avenues for future research. We further numerically investigate the most probable transition phase in a conceptual Arctic sea ice model due to Eisenman and Wettlaufer~\cite{eisenman2009nonlinear}, which has a strong seasonal forcing and possesses bistable situations. We use insights from our case study to conjecture that this model is in regime IV of Figure \ref{Fig:Example}, the fast forcing case. 

To simplify the notation, throughout the paper we use $\dot{x}$ to represent time derivative $dx/dt$ and subscript to denote partial derivative, e.g. $V_x(x,t) := \partial V/\partial x$. 

%
%
\section{Mathematical formulation}\label{sec:mathForm}
\subsection{Deterministic problem} \label{sec:probForm}
We first discuss the role of the parameters in the deterministic problem
\begin{align} 
\dot{x} = \frac{1}{\varepsilon}f(x,t) &=-\frac{1}{\varepsilon}V_x(x,t)\nonumber\\
&= \frac{1}{\varepsilon}\left(x-x^3 + \alpha + A \cos (2\pi t)\right),\label{eqn:ODE}
\end{align}
where $V(x,t)$ is a potential associated with the deterministic dynamics $f(x,t)/\varepsilon$, which is the drift term in \eqref{eqn:ModelSDE}:
\begin{equation}\label{eq:potential}
V(x,t)=\frac{x^4}{4}-\frac{x^2}{2}-\left(\alpha+A\cos\left(2\pi t\right)\right)x. 
\end{equation}
The parameter $A$ determines the forcing strength and $\alpha$ breaks a reflection symmetry of the potential. The timescale ratio $\varepsilon$ controls how (un)stable the periodic solutions $x^*$ of (\ref{eqn:ODE}) are; their Floquet multipliers, $\rho(x^*)$, depend exponentially on $1/\varepsilon$:
\begin{equation}\label{Eqn:FloquetMultiplier}
\rho(x^*)=\exp\left(\frac{1}{\varepsilon}\int_0^1 f_{x}(x^*(t),t)dt\right).
\end{equation}

\begin{figure}[t]
\centering
\begin{subfigure}[b]{0.4\textwidth}
\includegraphics[width=\textwidth]{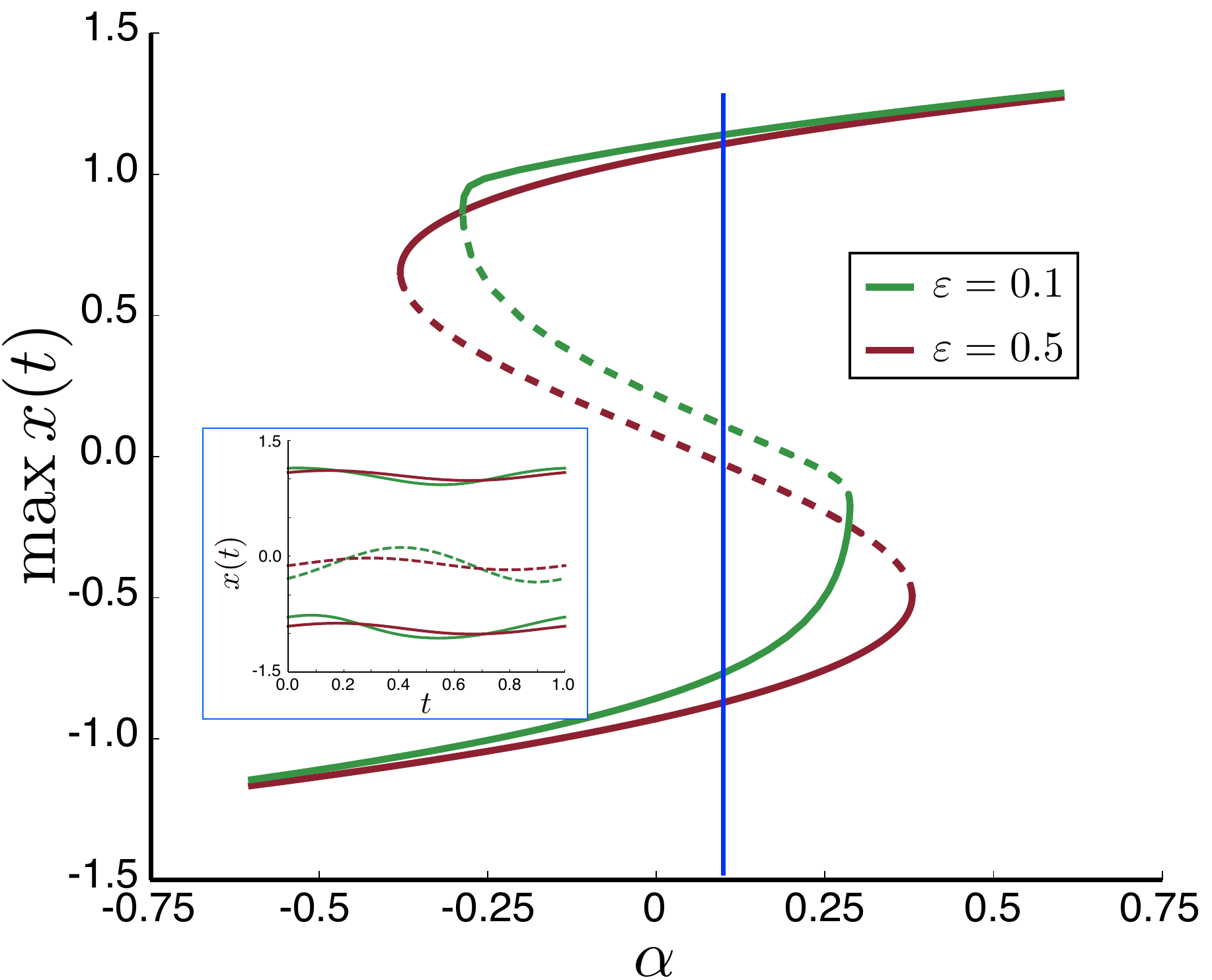}
\caption{}
\end{subfigure}
\begin{subfigure}[b]{0.38\textheight}
\includegraphics[width=\textwidth]{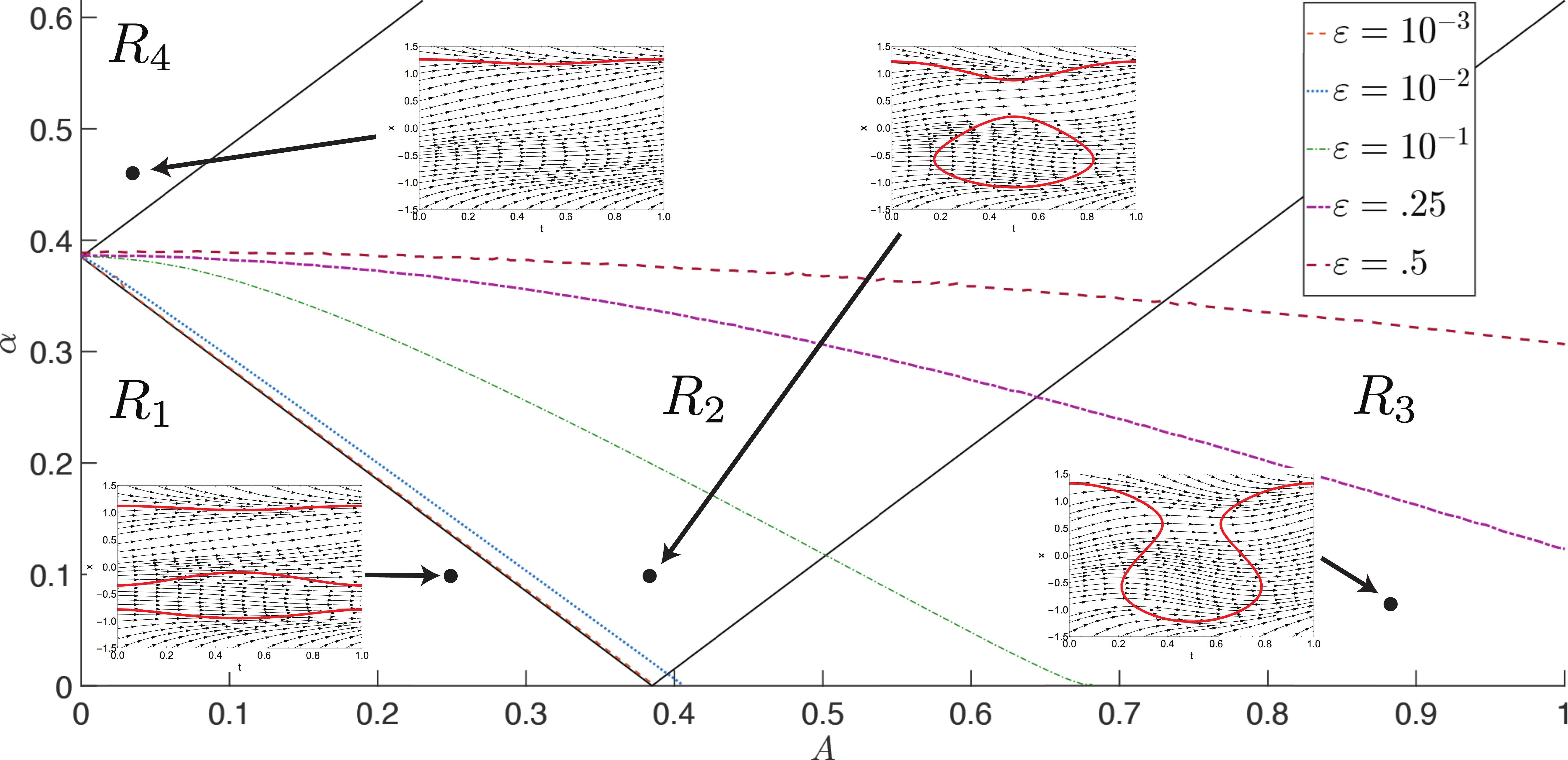}
\caption{}
\end{subfigure}

\caption{(a): Hysteretic bifurcation diagrams (in $\alpha$) associated with (\ref{eqn:ODE}) for two $\varepsilon$ values and $A=0.25$. The inset figure shows the solutions for $\alpha = 0.1$.  (b): Regimes $R_1-R_4$, in the ($A,\alpha$)-parameter plane, with qualitatively different flow geometry; inset figures illustrate the flow geometry for each, with the nullclines (curves where $f(x,t)=0$) overlaid in red. The dashed curves track the location of the right-most saddle node bifurcation in $\alpha$,  for various $\varepsilon$ values (the region with three periodic solutions is below these curves). }
\label{Fig:DynamicsDiagram}
\end{figure}

Figure \ref{Fig:DynamicsDiagram} summarizes the dependence of the deterministic dynamics of (\ref{eqn:ODE}) on the parameters. The phase diagram in Figure \ref{Fig:DynamicsDiagram}(b) shows that (\ref{eqn:ODE}) admits four qualitatively distinct parameter regions for $\alpha>0$; these are characterized by different configurations of the nullclines $f(x,t)=0$, in the $(t,x)$-phase plane. In this paper we focus on regions $R_1$ -- $R_3$ of the $(A, \alpha)$-parameter plane where, for a given $\varepsilon$, there are two stable periodic solutions. The location of the saddle-node bifurcations of Figure \ref{Fig:DynamicsDiagram}(a) depend on $\varepsilon$ as indicated. (There is only one stable periodic solution in $R_4$.) We give the description of the three parameter regions and their nullclines configurations as follows:
\begin{itemize}
    \item The nullclines in $R_1$, defined by $0<\alpha+A< 2/(3\sqrt{3})$, consist of three curves that partition extended phase space into four regions with vertical flow directions that alternate in sign. A standard Poincar\'e analysis proves existence of two stable periodic solutions for all $\varepsilon$ values in $R_1$.
    \item Region $R_2$, defined by $2/(3\sqrt{3})<\alpha+A$ and $|\alpha-A|<2/(3\sqrt{3})$, has two nullcline curves that divide the phase space into three regions.
    \item In $R_3$, defined by $0<\alpha<A-2/\sqrt{3}$, phase space is partitioned by the nullcline into two regions. 
\end{itemize}

Figure \ref{Fig:DynamicsDiagram}(b) highlights another $\varepsilon$-dependent feature of the dynamics: when $\varepsilon 
\ll 1$ in (\ref{eqn:ODE}) the periodic solutions  are well approximated by the nullclines $f(x,t)=0$ in $R_1$. As $\varepsilon$ increases, deviations between the periodic solutions and  the nullclines increase, leading to a pronounced intertwining of the nullclines with $x_l^*$ and $x_m^*$. This indicates that there can be intervals of time when the flow is expanding about a stable orbit  and/or contracting about the unstable one outside of the adiabatic regime.
These intervals become more pronounced
in $R_2$ and $R_3$, where, provided $\varepsilon$ is large enough, three periodic solutions may exist even though there are fewer than three curves comprising the nullcline; in these regions the double well structure of $V(x,t)$ is temporarily lost during certain phases of the forcing. 

\subsection{Path integral approach}\label{sec:pathIntegralForm}
We use the path integral formulation to identify a possible preferred phase for tipping. For $t_f>t_0\geq 0$, let $\mathcal{P}=\{x\in C([t_0,t_f];\mathbb{R}):x(t_0)=x_0 \text{ and } x(t_f)=x_f\}$ denote the set of paths
connecting a point $(t_0,x_0)$ on $x_l^*$ and a point $(t_f,x_f)$ on $x_u^*$. If $x\in \mathcal{P}$ is differentiable, then the probability of the tipping trajectories lying in an infinitesimally small tubular neighborhood of $x, P_\sigma[x]$, satisfies
\begin{align}\label{eq: transProb}
P_{\sigma}[x] \propto
\int_{\mathcal{P}}
&\exp\bigg[-\frac{1}{2\sigma^2}\bigg(\int_{t_0}^{t_f}\bigg(\dot{x}-\frac{1}{\varepsilon}f(x,\tau)\bigg)^2d\tau \nonumber\\
&+ \frac{\sigma^2}{\varepsilon}\int_{t_0}^{t_f} f_x(x,\tau) d\tau \bigg)\bigg] d_W\left[x\right],
\end{align}
where  $d_W\left[x\right]$ is the Wiener measure \cite{wiegel1986introduction, chaichian2001path}, and the proportionality reflects the missing normalization prefactor. Heuristically, this formulation can be derived by considering an equally spaced partition $t_0=t_1<t_2<\ldots <t_n=t_f$ of width $\Delta t$, calculating the probability that a realization passes through a sequence of intervals $[a_i,b_i]$ at time $t_i$, and taking the limit $|b_i-a_i|\rightarrow 0$ and $\Delta t\rightarrow 0$; derivations following this approach can be found in the physics literature \cite{wiegel1986introduction, chaichian2001path}. 

The probability of a tipping event lying in a small tubular neighborhood of a tipping path $x\in\mathcal{P}$ is then determined by integrating $P_\sigma[x]$ over this neighborhood. Similar to Laplace's method for the asymptotic expansion of exponential integrals, we expect as $\sigma\rightarrow 0$ that the dominant contribution to this integral is concentrated about local minimizers of the argument of the exponential in the definition of (\ref{eq: transProb}). With this as motivation, we define the \emph{most probable path} or \emph{optimal path} connecting points $(t_0,x_0)$ and $(t_f,x_f)$ to be minimizers of the Onsager-Machlup (OM) functional $I_{\sigma}: \mathcal{A}\mapsto \mathbb{R}$ given by
\begin{align}\label{eq:rateFunctional}
I_{\sigma}[x]:=&\int_{t_0}^{t_f}L(t, x,\dot{x})dt\nonumber \\
=&\underbrace{\int_{t_0}^{t_f} \left(\dot{x}-\frac{1}{\varepsilon}f(x,t)\right)^2
dt}_{I_{\text{FW}}} +\frac{\sigma^2}{\varepsilon}\underbrace{ \int_{t_0}^{t_f} f_x(x,t)dt}_{I_\text{OM2}},
\end{align}
which can be expressed as the sum of the Freidlin--Wentzell (FW) rate functional $I_\text{FW}$ derived by Freidlin and Wentzell \cite{freidlin2012random} and an additional integral $I_\text{OM2}$. Here the admissible set is $\mathcal{A}:=\{x\in H^1([t_0,t_f];\mathbb{R}): x(t_0)=x_0, x(t_f)=x_f\}$ and $L:([t_0,t_f],\mathbb{R},\mathbb{R})\mapsto \mathbb{R}$ denotes the Lagrangian.

Since $L(t, x,\dot{x})$ is convex in $\dot{x}$, to prove the existence of a minimum in $\mathcal{A}$ it is sufficient to show that $I_{\sigma}$ is coercive with respect to the $H^1$ norm. The next theorem establishes this fact (as $\|\dot{x}\|_{L^2}\rightarrow \infty, I_\sigma\rightarrow \infty$); the existence of a minimizer is then a corollary following standard techniques from the direct method of the calculus of variations \cite{jost1998calculus}.
\begin{theorem}\label{Thm:Coercivity}
There exists $M\in \mathbb{R}$ such that for all $x\in \mathcal{A}$ and all $\sigma\in [0,1]$
\begin{equation*}
I_{\sigma}[x]> \|\dot{x}\|_{L^2}+M.
\end{equation*}
\end{theorem}
\begin{proof}
Let $x\in \mathcal{A}$, $\sigma\in [0,1]$, and $\Delta V=V(x_f,t_f)-V(x_0,t_0)$, where the potential $V$ is given by \eqref{eq:potential}. Expanding the integrand of $I_\text{FW}$ and integrating $\dot{x}\cos(2\pi t)$ by parts, we find 
\begin{align*}
&I_{\sigma}[x]=\|\dot{x}\|_{L^2}+2\varepsilon^{-1}\Delta V\\
+&\frac{1}{\varepsilon}\int_{t_0}^{t_f}\left(4\pi Ax(t)\sin(2\pi t)+\frac{f(x(t),t)^2}{\varepsilon}+\sigma^2 f_x(x(t),t)\right)dt.
\end{align*}
Observing that $\sigma^2f_x(x(t),t)=\sigma^2 (1-3x(t)^2)\geq -3x(t)^2$ for $\sigma\in[0,1]$, it follows that $I_{\sigma}[x]\geq \|\dot{x}\|_{L^2}+2\varepsilon^{-1}\Delta V+\int_{t_0}^{t_f}p(x(t),t)dt$, where $p(x,t)$ is a sixth degree polynomial in $x$ with coefficients independent of $\sigma$ and satisfying $\lim_{x\rightarrow \pm \infty}p(x,t)=\infty$. If $m$ denotes the absolute minimum value of $p(x,t)$ on the interval $[t_0,t_f]$, and $M=2\varepsilon^{-1}\Delta V+(t_f-t_0)m$, then
\begin{equation*}
I_{\sigma}[x]> \|\dot{x}\|_{L^2}^2+2\varepsilon^{-1}\Delta V+(t_f-t_0)m=\|\dot{x}\|_{L^2}+M.
\end{equation*}
\end{proof}
\begin{corollary}\label{Cor:Existence}
For all $\sigma\geq 0$ there exists $x^*\in \mathcal{A}$ such that for all $x\in \mathcal{A}$,
$I_{\sigma}[x^*]\leq I_{\sigma}[x]$.
\end{corollary}


Local minimizers of (\ref{eq:rateFunctional}) satisfy the Euler-Lagrange equations: 
\begin{equation}\label{eqn:EL}
\begin{cases}
\ddot{x} = \frac{1}{\varepsilon^2}f(x,t)f_x(x,t) + \frac{\sigma^2}{2\varepsilon}f_{xx}(x,t) +\frac{1}{\varepsilon} f_t(x,t)\\
x(t_0)=x_0\\
x(t_f)=x_f
\end{cases}.
\end{equation}
These equations constitute a second order nonlinear boundary value problem in $x$ on the domain $[t_0,t_f]$, and can be rewritten in Hamiltonian form, for
time-varying Hamiltonian function
\begin{equation}
H(x,\Psi,t)=\frac{\Psi^2}{2}+\frac{1}{\varepsilon}f(x,t)\Psi-\frac{\sigma^2}{2\varepsilon}f_{x}(x(t),t),
\end{equation} 
as:
\begin{equation}\label{Eq:Hamiltonian}
\begin{cases}
\dot{x}=\Psi+\frac{1}{\varepsilon}f(x,t)\\
\dot{\Psi}=-\frac{1}{\varepsilon}f_x(x,t)\Psi+\frac{\sigma^2}{2\varepsilon}f_{xx}(x,t).
\end{cases}
\end{equation}
Here the `momentum' $\Psi=\dot{x}-\varepsilon^{-1}f(x,t)$ measures the deviation from the deterministic flow.
The FW functional in (\ref{eq:rateFunctional}) can then be expressed in terms of $\Psi$ as $I_{\text{FW}}[x(t)]=\int_{t_0}^{t_f} \left(\Psi(x,t)\right)^2dt$.

The path integral formulation assumes that the noise strength $\sigma$ is small. If $\sigma$ is sufficiently small compared to the timscale ratio $\varepsilon$, such that $\sigma^2/\varepsilon\ll 1$,  the naive regular perturbation in $\sigma$ of (\ref{Eq:Hamiltonian}), at leading order in $\sigma$, is:
 \begin{equation}\label{eqn:smallsigma}
 \dot{x}=\Psi+\frac{1}{\varepsilon}f(x,t), \quad
\dot{\Psi}=-\frac{1}{\varepsilon}f_x(x,t)\Psi.
\end{equation}
In this system $\Psi=0$ forms an invariant submanifold foliated by solutions of the deterministic dynamics $\dot{x}=\varepsilon^{-1}f(x,t)$. Thus, for $\sigma\rightarrow 0$, optimal paths are well-approximated by heteroclinic orbits of (\ref{Eq:Hamiltonian}) connecting deterministic solutions \cite{dykman2001activated}. Moreover, (\ref{eqn:smallsigma}) is the Hamiltonian form of the Euler-Lagrange equations for the FW rate functional $I_\text{FW}$. Since $I_\text{OM2}$ is independent of $\dot{x}$, minimizers of $I_{\sigma}$ will converge uniformly as $\sigma \rightarrow 0$ to a minimizer of $I_\text{FW}$; see Appendix \ref{sec:TheoremApend}.

%
%
\section{Non-adiabatic forcing regime: $\varepsilon = {\cal O}(1)$}\label{sec:noiseDriftBalance}
The remainder of the paper addresses the intermediate regime I in Figure \ref{Fig:Example} 
where $\varepsilon$ is ${\cal O}(1)$. In this case, there is no small parameter to exploit to simplify the analysis.
We investigate the contribution $I_\text{OM2}$ to the functional \eqref{eq:rateFunctional} in determining optimal paths, which are then validated by Monte Carlo simulations. This section concludes with a parameter study of \eqref{eqn:ModelSDE}, based on the path integral method using the Onsager-Machlup functional (\ref{eq:rateFunctional}).

\subsection{Gradient flow}\label{sec:gradFlow}
To solve the second order boundary value problem (\ref{eqn:EL}) we introduce a gradient flow $u_s=-\frac{\delta I_{\sigma}[u]}{\delta u}$ associated with $I_\sigma$. Specifically, introducing the artificial ``time'' $s\geq 0$ we consider solutions $u(s,t):\mathbb{R}^+\times [t_0,t_f]\mapsto \mathbb{R}$ to the following partial differential equation:
\begin{equation}\label{eq: GradFlow}
\begin{cases}
u_{s}=u_{tt} - \frac{1}{\varepsilon^2}f(u,t)f_u(u,t) - \frac{\sigma^2}{2\varepsilon}f_{uu}(u,t) - \frac{1}{\varepsilon}f_t(u,t), \\
u(s,t_0) = x_0, u(s,t_f) = x_f, \quad \text{for }s\geq 0\\
u(0,t)=x_g(t), \quad \text{for } t\in[t_0,t_f]
\end{cases},
\end{equation}
where $x_g(t)$ is a sufficiently smooth initial guess of the most probable path. To numerically solve (\ref{eq: GradFlow}) we use MATLAB's built in parabolic initial-boundary value solver {\sl{pdepe}} \cite{skeel1990method}. 

Stationary solutions of \eqref{eq: GradFlow} correspond to solutions of the Euler-Lagrange equations (\ref{eqn:EL}). If we let $u(s,t)$ denote a solution to (\ref{eq: GradFlow}), it follows that
\begin{equation}
\frac{d}{ds}I_\sigma(u(s,\cdot)) =\int_{t_0}^{t_f} \frac{\delta I_{\sigma}}{\delta u}\frac{\partial u}{\partial s}dt=-\int_{t_0}^{t_f} \left(\frac{\partial u}{\partial s}\right)^2dt\leq 0,
\end{equation}
and hence $I_{\sigma}[u(s,\cdot)]$ is a monotonically decreasing function in $s$. Furthermore, by Theorem \ref{Thm:Coercivity}, $I_{\sigma}[u(s,\cdot)]$ is bounded from below and thus the pointwise limit $\lim_{s\rightarrow \infty}I_{\sigma}[u(s,\cdot)]$ exists and $\lim_{s\rightarrow \infty} \frac{d}{ds}I_{\sigma}[u(s,\cdot)]=0$. $I_{\sigma}$ is a Lyapunov function on $\mathcal{A}$ and, by Corollary \ref{Cor:Existence}, $\lim_{s \rightarrow \infty}u(s,t)=x(t)$, where $x(t)$ is a local minimizer of the OM functional and hence a solution to the Euler-Lagrange equations \cite{robinson2001infinite}. Convergence to a steady state is assessed by the value of $\left|\frac{\delta I}{\delta u}\right|$, i.e. when the Euler-Lagrange equations are approximately satisfied.

\begin{figure}[ht]
    \centering
    \begin{subfigure}[b]{0.235\textwidth}
        \includegraphics[width=\textwidth]{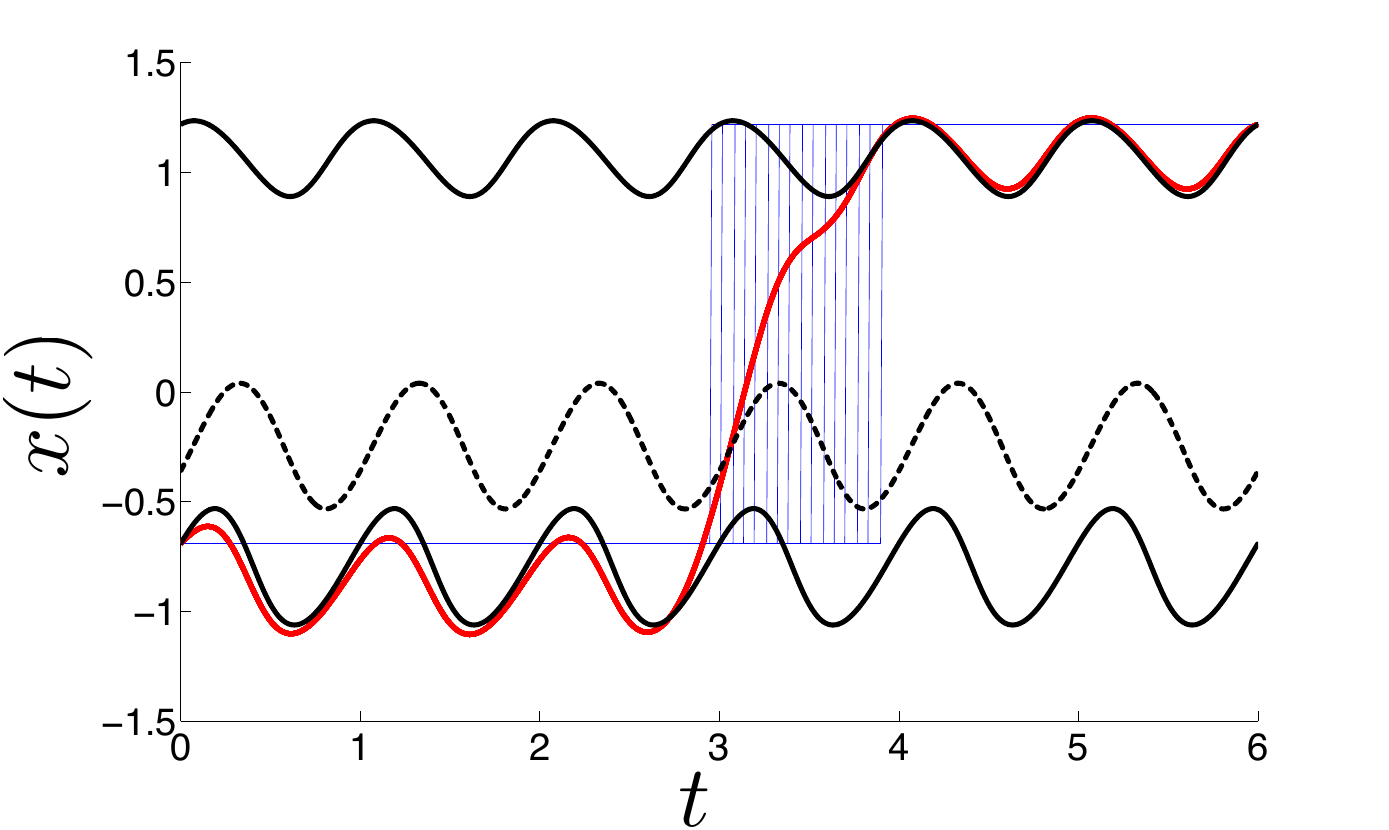}
        \caption{}
    \end{subfigure}    
    \begin{subfigure}[b]{0.235\textwidth}
        \includegraphics[width=\textwidth]{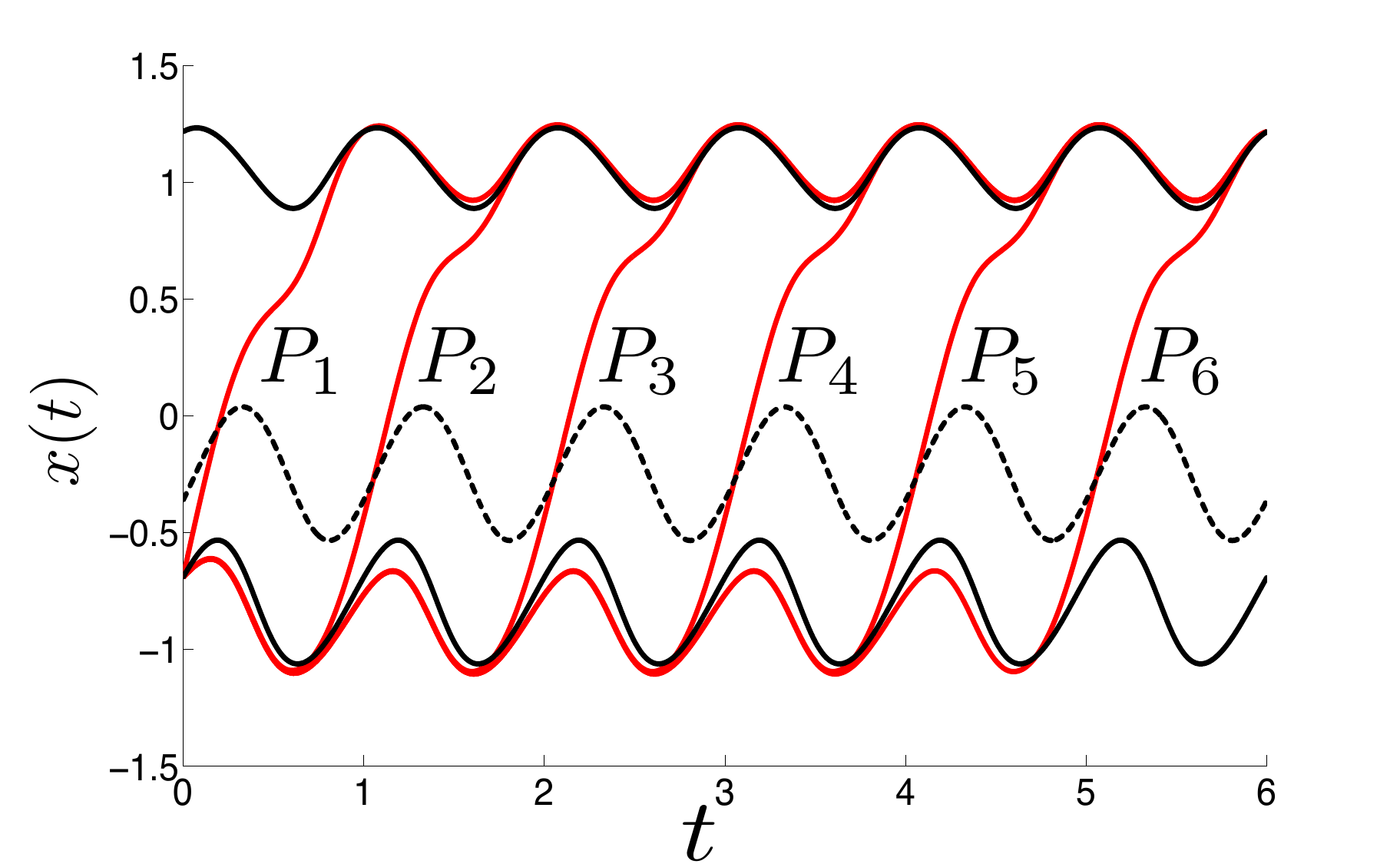}
        \caption{}
    \end{subfigure}
    \caption{
Results for $(\varepsilon,A,\alpha,\sigma) = (0.25, 0.5,  0.2,  0.4$).
(a): Examples of $16$ initial guesses (blue) evolving under the gradient flow to the same transition path (red). (b): Stationary solutions of \eqref{eq: GradFlow} for $64$ initial guesses of the form \eqref{eq:initialGuess} with  uniformly spaced  $t_j\in[0,6]$. Note that $P_3, ..., P_6$ are translates, by successive forcing periods, of $P_2$.}
\label{fig: varyTm}
\end{figure}

For the gradient flow we impose a starting position at $x_0 = x_l^*(t_0)$ and a final position at $x_f = x_u^*(t_f)$. Figure~\ref{fig: varyTm}  presents results for six forcing periods between $t_0$ and $t_f$. The initial condition for the gradient flow is chosen as a piecewise constant function 
\begin{equation}\label{eq:initialGuess}
u(0,t)=
\begin{cases}
x_0, \quad t\in[t_0,t_j]\\
x_f,\quad t\in (t_j,t_f],
\end{cases}
\end{equation}
where $t_j$ is varied on a uniform grid between $t_0$ and $t_f$. A few examples of such initial conditions $u(0,t)$ are shown as blue curves in Figure \ref{fig: varyTm}(a), where $t_j$ takes values uniformly in $[3,4]$ and all of the initial solutions converge to the same stationary solution represented by the red curve. In this example we find six different stationary solutions $\{P_i: i = 1...6\}$, each of which represents a different local minimizer of $I_\sigma$ given by \eqref{eq:rateFunctional}; see Figure \ref{fig: varyTm}(b). The paths $\{P_2,...,P_6\}$ are the same up to translation by a multiple of the forcing period. $P_1$ is also a local minimizer of $I_\sigma$, but its form depends on $t_0$, whereas the remaining optimal paths are insensitive to the choice of $t_0$ and $t_f$; see Appendix \ref{sec:independenc} for details.

\begin{figure}[ht]
        \includegraphics[width=0.4\textwidth]{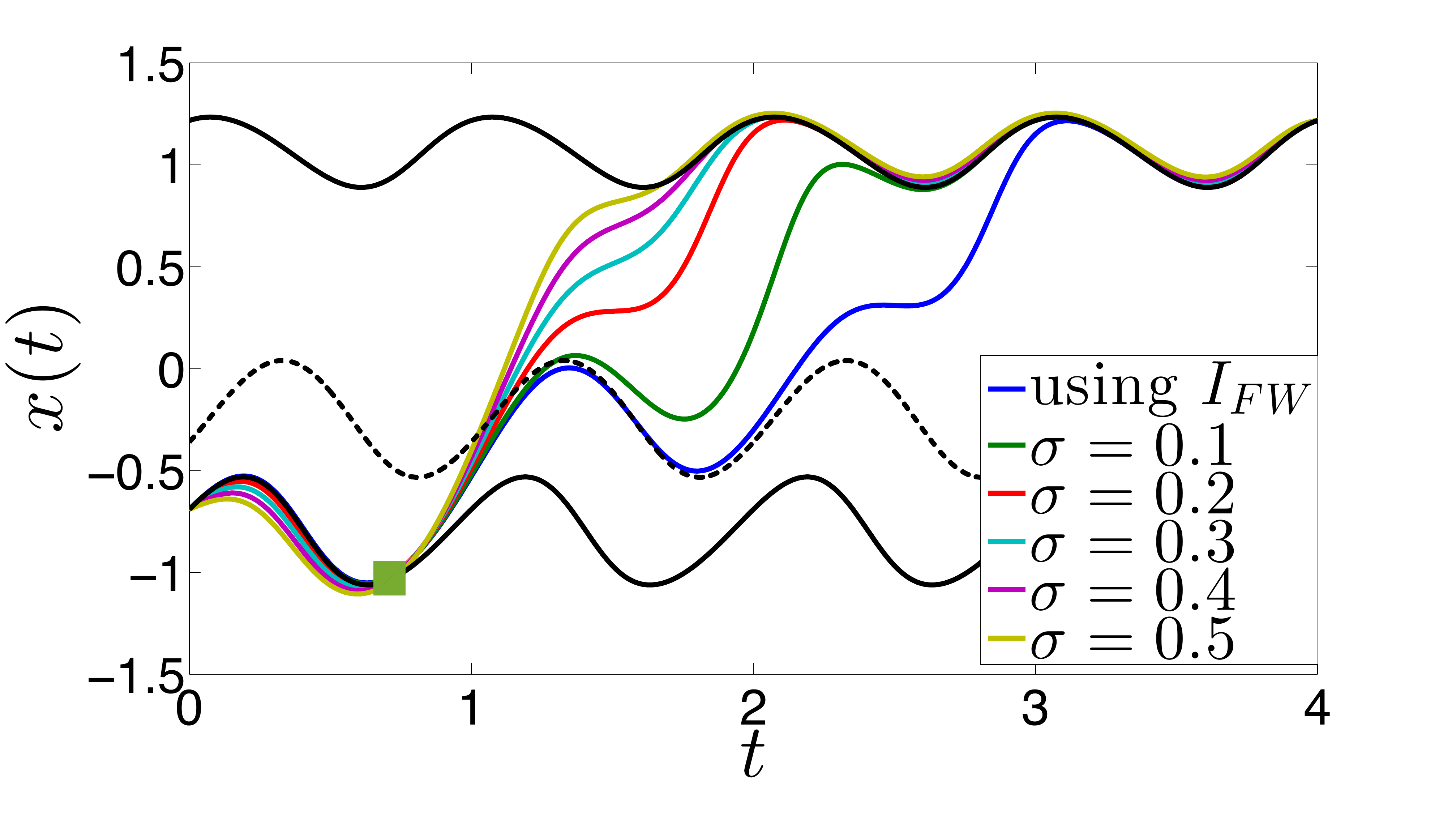}
\caption{Stationary solutions to \eqref{eq: GradFlow} for $(\varepsilon,A,\alpha) = (0.25, 0.5, 0.2)$ and various $\sigma$, where ``using $I_{FW}$" corresponds to setting $\sigma=0$ in the functional \eqref{eq:rateFunctional}. The green square marks where paths cross $x_l^*$. }\label{fig: FWvsOM}
\end{figure}

Figure \ref{fig: FWvsOM} compares optimal paths, obtained from the gradient flow \eqref{eq: GradFlow}, for various values of $\sigma$. This includes the setting $\sigma=0$, which corresponds to using the FW functional. This example enforces boundary conditions so that all optimal transition paths start at the same initial point on $x_l^*$ and terminate four periods later on $x_u^*$. We make the following observations:
\begin{enumerate}
\item Before crossing $x_m^*$, the optimal paths are concentrated along a similar trajectory, and they all leave $x_l^*$ around the same phase of the forcing, marked by the green square. However, after crossing $x_m^*$, the optimal paths spread out in a way that is strongly dependent on the noise intensity.  

\item The minimizer of $I_\text{FW}$ stays near $x_m^*$ for a full cycle before it exits to $x_u^*$. This is in contrast to the minimizers of $I_\sigma$ where the path may transition quickly to $x_u^*$. 

\item Optimal paths for different $\sigma$ cross $x_m^*$ at different times, and quickly relax to follow the deterministic flow after crossing. 

\item For larger  $\sigma$, there is a small gap between the optimal path and the stable periodic solutions before and after the transition.  This points to a possible failure of the minimizers of the OM functional to adequately approximate the most probable path for larger $\sigma$ since we  expect the system to track the deterministic dynamics there.
\end{enumerate}

\subsection{Comparison with Monte Carlo simulations}\label{sec:MCcomparison}
\begin{figure}[ht]
    \centering
        \includegraphics[width=0.5\textwidth]{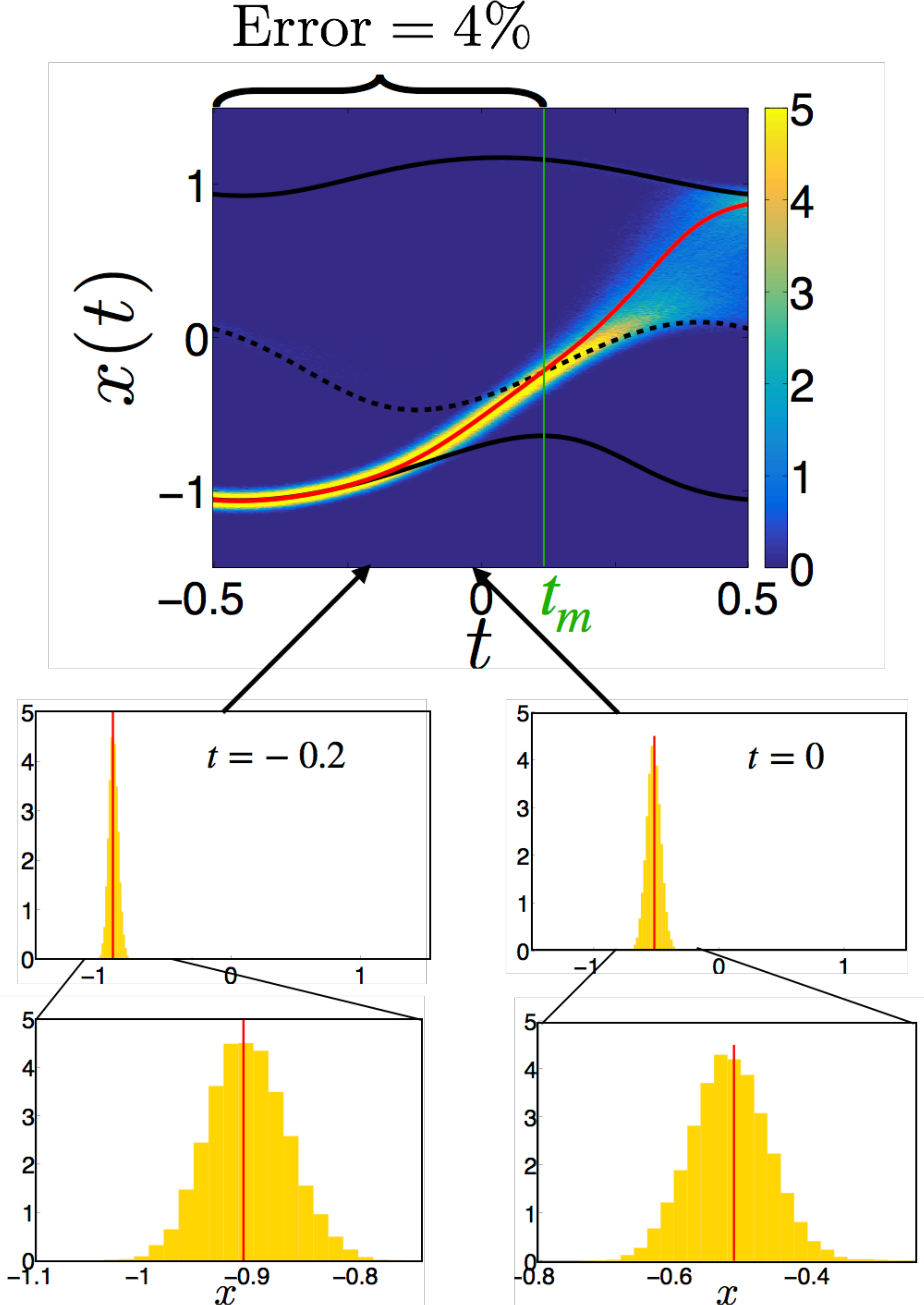}
    \caption{Comparison of optimal paths with time-evolved distribution obtained via Monte Carlo simulations at $(\varepsilon, A, \alpha, \sigma) = (0.1, 0.3, 0.15, 0.2)$. The mean absolute percentage error between the mode of the time-evolved distribution and the optimal path for $t\in [-0.5, t_m]$ is $4\%$, where $t_m \approx 0.1$. Two examples of normalized distributions of $x$ are shown in yellow at $t=-0.2$ and $t=0$, and compared to the value of the optimal path position, indicated by a red vertical line. The bottom two figures are zoomed-in distributions in $x$ at $t = -0.2$ and $t=0$. $E = 8\times 10^{-6}$ defined as \eqref{eqn: meansquarederror} in the convergence test. There are $N = 12772$ tipping samples; $M=50$ bins used for estimating the distribution of $x\in [-1.5, 1.5]$, and the number of time steps in each sample is $K = 10^4$.}
\label{fig: compMC}
\end{figure}

We numerically validate the path integral method by comparing optimal paths with the results of Monte Carlo simulations. One example is demonstrated in Figure \ref{fig: compMC} for $(\varepsilon, A, \alpha, \sigma) = (0.1, 0.3, 0.15, 0.2)$.
The reported mean squared error associated with halving the number of samples quantifies our convergence test of the Monte Carlo results to the distribution. It is defined as,
\begin{equation}\label{eqn: meansquarederror}
E = \sum_{k = 1}^{K}\sum_{m = 1}^M \frac{\left(h_{N}^{(m,k)} - h_{N/2}^{(m,k)}\right)^2}{MK},
\end{equation}
where $h_{N}^{(m,k)}$ is the fraction of samples in each bin $m$;  see Appendix \ref{sec:MCappendix} for more details about the convergence test and six comparisons for different values of $(\varepsilon, \sigma)$.

Figure \ref{fig: compMC} compares results obtained with the path integral method to those from Monte Carlo simulations for parameters $(\varepsilon, A, \alpha, \sigma) = (0.1, 0.3, 0.15, 0.2)$. This figure shows the distribution of the sample paths on the (unwrapped) cylinder and two examples of distributions in $x$ at $t = -0.2$ and $0$, compared with values of optimal paths (marked by red lines) at those times. These distributions are fairly concentrated as shown in the middle row of Figure \ref{fig: compMC}. 
We zoom in near the peaks of the distributions, in the bottom row, to better reveal the agreement between the optimal path and the Monte Carlo simulation results. Moreover, we notice that distributions broaden after $t_m$, the time when paths cross the unstable periodic solutions $x_m^*$, after which paths will simply follow the deterministic flow to $x_u^*$.


\begin{table}[h]
\begin{tabular}{| c | c | c | c |}
  \hline
   & $\sigma$ & N & MAPE \\
   \hline
   & 0.2 & $49002$ & $5.0\%$ \\
    $\varepsilon = 0.25$ & 0.25& $46010$ & $5.6\%$ \\
       & 0.3& $21183$ & $6.3\%$ \\
       \hline
      & 0.2&$10662$ & $6.0\%$ \\
    $\varepsilon = 0.4$ & 0.25& $24882$ & $6.8\%$ \\
       & 0.3& $48160$ & $8.3\%$ \\
       \hline
\end{tabular}
       \caption{Results of comparison for $\varepsilon = 0.25$ and $\varepsilon = 0.4$ with noise intensity $\sigma = 0.2, 0.25$ and $0.3$. $(\alpha, A)$ are fixed at $(0.15, 0.7)$. Illustration of the comparison can be found in Figure \ref{fig: more_MC_comparison} in Appendix \ref{sec:MCappendix}. $N$ is the number of tipping samples used in the time-evolved distribution; MAPE (defined in \eqref{eqn: MAPE}) is the mean absolute percentage error between the mode of the distribution and the optimal path for $t\in[-0.5, t_m]$, where $t_m\in[0, 0.3]$ approximately for the parameters we take. }\label{table: MCcomparison}
\end{table}

To quantify the discrepancy between optimal path and Monte Carlo simulation results, we compute the mean absolute percentage error (MAPE) between the mode of time-evolved distributions and optimal paths for $t\in [-0.5, t_m]$, i.e. before crossing $x_m^*$. 
The MAPE is calculated as 
\begin{equation}\label{eqn: MAPE}
MAPE = \frac{100\%}{k}\sum_{i=1}^{k}\left|\frac{(m_i - o_i)}{\bar{m}}\right|, 
\end{equation}
where $m_i$ denotes the mode of the distribution at each time step $t_i$, $o_i$ is the value of optimal path in $x$ at $t_i$, $\bar{m}=\sum_{i=1}^k m_i/k$ is the average of the mode of the distribution over $t\in [-0.5, t_m]$, and $k$ represents the total number of time steps in $t\in [-0.5, t_m]$.  More comparison results for different values of $\varepsilon$ and $\sigma$ are summarized in Table \ref{table: MCcomparison} and plots are presented in Appendix \ref{sec:MCappendix}. In these comparisons, $\alpha$ and $A$ are fixed at $0.15$ and $0.7$, respectively. Results from the Monte Carlo comparison suggest that the optimal path obtained from the OM functional gives a good estimate of the mode of the distribution prior to crossing $x_m^*$ in this intermediate forcing regime.

\subsection{Parameter study}\label{sec:parameterStudy}
In this section we use the optimal path method to investigate the intermediate forcing regime I of Figure \ref{Fig:Example}, and explore features of the deterministic problem which influence tipping phase. Figure \ref{fig: FWvsOM} demonstrated that transitions from the lower stable solution, $x_l^*$, to the unstable solution, $x_m^*$, may be insensitive to changes in $\sigma$, leading to a well-defined transition path associated with this part of the transition. Here we explore the transitions between $x_l^*$ and $x_m^*$ more systematically for parameters in the nullcline regions $R_2$ and $R_3$, which were introduced in Figure \ref{Fig:DynamicsDiagram}(b) for which the number of nullcline curves corresponds to $2$ and $1$, respectively.

Figure \ref{fig:fixEpsVarySigma} summarizes results of a parameter study for $\alpha$ fixed at $0.1$, $\sigma\in [0.1, 0.4]$ and $\varepsilon = 0.25$ and $0.4$. $A$ is $0.4$ in the left column and $0.7$ in the right column, each associated with the different nullcline configurations that are in $R_2$ and $R_3$ of Figure \ref{Fig:DynamicsDiagram}(b). We make the following observations:

\begin{enumerate}
\item For all four deterministic parameters sets, the optimal paths leave $x_l^*$ near when a nullcline crosses $x_l^*$, indicating a time when the flow changes direction from contracting to expanding near $x_l^*$. 
\item  The transition between $x_l^*$ and $x_m^*$ is concentrated around the same trajectory and varies little with the change of $\sigma$. Moreover, the transitions from $x_l^*$ to $x_m^*$ all happen in the gap between the nullclines where the deterministic flow is expanding from $x_l^*$, independent of noise intensity. 
\item For $\varepsilon = 0.25$, the optimal paths arrive at $x_m^*$ around the same time for different values of $\sigma$. However for $\varepsilon = 0.4$, the times $t_m$ differ by quite a bit. This is because for larger values of $\varepsilon$ the flow is weaker and requires larger noise intensity to see a comparably fast transition to that obtained with $\varepsilon=0.25$.
\item Upon crossing the unstable periodic solution $x_m^*$, the optimal paths spread out and follow the deterministic flow to reach the upper stable periodic solution $x_u^*$. 
\end{enumerate}

\begin{figure}[h]
\centering
\includegraphics[width=0.5\textwidth]{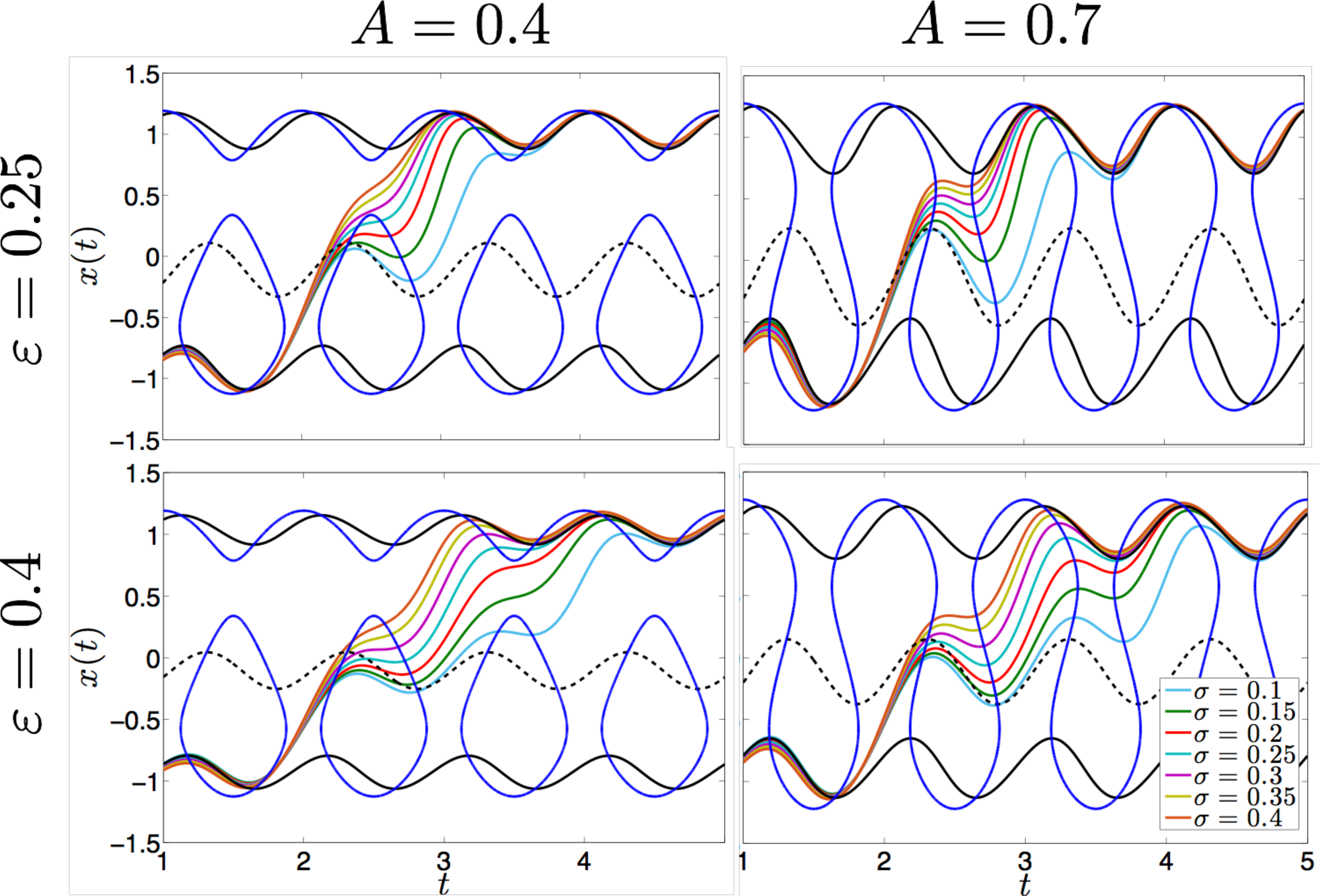}
\caption{Summary of optimal paths obtained using the path integral method for various noise intensities $\sigma \in [0.1,0.4]$. Nullclines of the deterministic system \eqref{eqn:ODE} are represented by blue curves. Black curves represent the periodic solutions of \eqref{eqn:ODE}. $\alpha$ is fixed at $0.1$. $\varepsilon = 0.25$ in the first row and $0.4$ in the second row. $A = 0.4$ in the first column and $0.7$ in the second column.}
\label{fig:fixEpsVarySigma}
\end{figure}

\begin{figure}[h]
\includegraphics[width=0.5\textwidth]{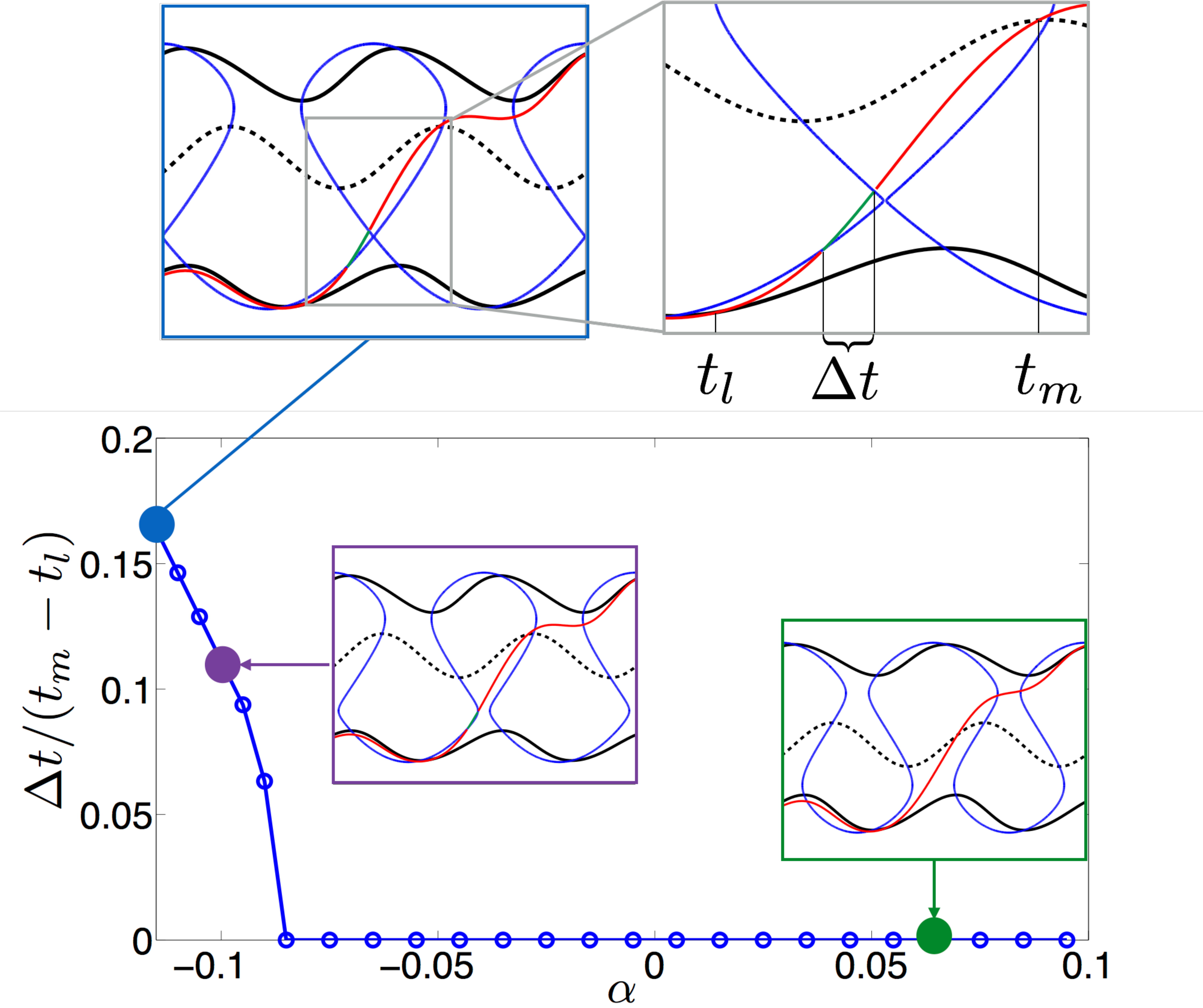}
\caption{$\Delta t/(t_m-t_l)$ as a function of $\alpha$, where $\Delta t$ is the time spent in regions of downward flow, denoted by green curves. The insets represent the optimal paths plotted with the nullclines and the periodic solutions for $\alpha = -0.135, -0.1$ and $0.06$. The rest of the parameters are fixed and are taken as $(A, \varepsilon, \sigma) = (0.5, 0.25, 0.4)$.}
\label{fig:varyAlpha}
\end{figure}
We propose, based on the explorations summarized by Figure \ref{fig:fixEpsVarySigma}, that a good deterministic predictor of the tipping phase $t_l$ is the point where the flow near $x_l^*$ changes direction from contracting to expanding. This point is associated with an appropriate nullcline crossing the lower stable orbit $x_l^*$.
Moreover, for the flow geometries of regions $R_2$ and $R_3$, there is a gap between the nullclines; we conjecture that this gap serves as a passageway through which the most likely path squeezes. To further explore this claim, we focus on parameter sets in $R_3$ where there is a single nullcline, i.e. as in the cases of the right column in Figure \ref{fig:fixEpsVarySigma}. We compute the fraction of the transition time from $x_l^*$ to $x_m^*$ which lies above the nullcline (`outside the gap') as a function of $\alpha$ which controls the width of the gap; see Figure \ref{fig:varyAlpha}. For $\alpha$ small (i.e. $\alpha$ near $-0.135$, the boundary between $R_2$ and $R_3$), the gap becomes very narrow, yet more than $80\%$ of the transition from $x_l^*$ to $x_m^*$ occurs `inside the gap' below the nullcline, indicating the flow is upward. For $\alpha> -0.095$ and below the saddle-node bifurcation, the entire transition from $x_l^*$ to $x_m^*$ is in this gap. 

\section{Discussion}\label{sec:Discussion} 
In this paper we explored a notion of a ``tipping phase'' for a periodically forced Langevin equation \eqref{eqn:ModelSDE}, and extended the deterministic conditions that set it from the adiabatic to the non-adiabatic regime. Previous work in the adiabatic regime shows that the most likely transition occurs when the potential barrier height achieves its minimum \cite{berglund2002metastability}, which for \eqref{eqn:ODE} is also the time when the separation between stable and unstable orbits is a minimum. Outside of the adiabatic regime, the two orbits exhibit roughly constant separation over time and the potential may lose its double well structure over an interval of time. 

As a tool in our study we used the path integral formulation to compute the optimal paths between two stable solutions. Specifically, we used the Onsager-Machlup functional whose minimizers we defined as optimal transition paths. The OM functional is the sum of the Freidlin-Wentzell rate functional and a second term that scales with $\sigma^2/\varepsilon$, the latter of which controls the ratio of the noise to drift strength. We argued, and showed, that this second term in the OM functional can be significant in regime I in Figure \ref{Fig:Example} when $\varepsilon = \mathcal{O}(1)$ and $\sigma\in (0,1)$, and in particular when the timescale ratio $\varepsilon$ and noise intensity $\sigma$ are comparable. 

Our findings in Section \ref{sec:parameterStudy} on the parameter study suggest that the time when the flow changes from contracting to expanding about the starting stable solution is a robust estimator for the tipping phase, which we defined to be $t_l$, the time when the path leaves $x_l^*$. 
In some sense this can be seen as an extension of the adiabatic regime results in that the transition takes place during the phase of the forcing when the barrier height vanishes (effectively, when it is minimal), but our numerical study highlights several key differences. Notably, the barrier height analogy breaks down in our case, as tipping starts occurring during a phase when only a single minimum to the associated potential exits. More importantly, while in the adiabatic regime tipping is nearly instantaneous (on time scale of forcing period) at the time point when the minimal potential barrier height is reached, in our case of an intermediate forcing regime, tipping takes much longer, making it necessary to define a tipping time more carefully. Our work suggests this estimator, $t_l$, and shows that long noisy paths attempt to make the transition by squeezing and maneuvering between portions of the nullclines.

\subsection{Future directions}\label{sec:Future}
We conclude with some future avenues of research that come out of this work and a discussion of some of the limitations of our approach that suggest open problems.
\paragraph{Bifurcations in the Hamiltonian system}{The Hamiltonian system (\ref{Eq:Hamiltonian}) admits several periodic solutions $\left(x^H(t), \Psi^H(t)\right)$ and has a rich bifurcation structure in $\sigma$ and $\varepsilon$. For example, when $\sigma = 0$ and for typical parameters in our studies, there exist five periodic solutions
 of (\ref{Eq:Hamiltonian}); three of these correspond to the solutions $x^*(t)$ of (\ref{eqn:ODE}) and are given by $\left(x^H(t), \Psi^H(t)\right) = \left(x^*(t),0\right)$, while the other two are periodic solutions for which $\Psi^H \neq 0$ (e.g.\ at $A=0$, they satisfy $x = \pm \sqrt{1/3}$). We conducted a preliminary bifurcation analysis of (\ref{Eq:Hamiltonian}) using AUTO-07P \cite{doedel1998auto} and found that qualitative differences in the 
tipping trajectories in Figure \ref{fig:fixEpsVarySigma}(d-f) roughly align with saddle-node bifurcations of the Hamiltonian dynamics under changes of $\sigma$ and $\varepsilon$. It would be interesting to explore whether the bifurcations in the Hamiltonian system can serve as a method for partitioning a $(\sigma,\varepsilon)$ phase diagram into qualitatively different tipping behaviors.}

\paragraph{Higher dimensions}{Our conclusions regarding the way  the tipping phase  aligns with a nullcline crossing are inherently tied to the low dimensionality of the problem we considered. Specifically, we considered a one degree of freedom, periodically forced system, for which an (extended) $1+1$ dimensional phase plane analysis allowed us to gain insights. While our approach may not be limited to 1-D systems, since we could base an investigation on the OM functional in higher dimensions, the natural analog of the nullclines is unclear. Are there situations where a separator of expanding and contracting regions could be identified, and could allow us to determine a preferred tipping phase based on deterministic properties? 

\paragraph{Limitations of the path integral approach}{The benefit of the path integral approach is that it  efficiently calculates optimal paths without performing Monte Carlo simulations. However, with this approach we lose higher order information that is encoded in the distribution of tipping events. In contrast, in the Freidlin--Wentzell theory of large deviations these quantities can be recovered from the so called quasi-potential in the limit of vanishing noise strength \cite{freidlin2012random,ren2004minimum}. A natural and open question is how to use the OM functional to obtain similar quantities in an asymptotic regime in which $\varepsilon=\mathcal{O}(1)$ and noise is small (noise-drift balanced) or $\varepsilon \ll 1$ (adiabatic). This may be important in the context of the question we pose about a tipping phase for the following reason. If the distribution of  tipping events is too broad then there is no well defined phase. It only gains meaning when the distribution is narrow compared to the period of the forcing. We expect the distribution will narrow with decreasing in $\sigma$, as illustrated in the contrasting histograms Figure \ref{fig: histogramsVarySigma}, which were computed with different noise intensities.
\paragraph{Application}{Our investigation arose, originally, in the setting of tipping points in a bistable conceptual model of Arctic sea ice, which has strong seasonal forcing. 
We wondered whether there is a  season of the year when this important component of the climate system is most vulnerable to tipping from its current state of perennial ice cover to an alternative state where the Arctic is, year-round, ice-free. 

\begin{figure}[ht]
\centering
\begin{subfigure}[b]{0.4\textwidth}
\includegraphics[width=\textwidth]{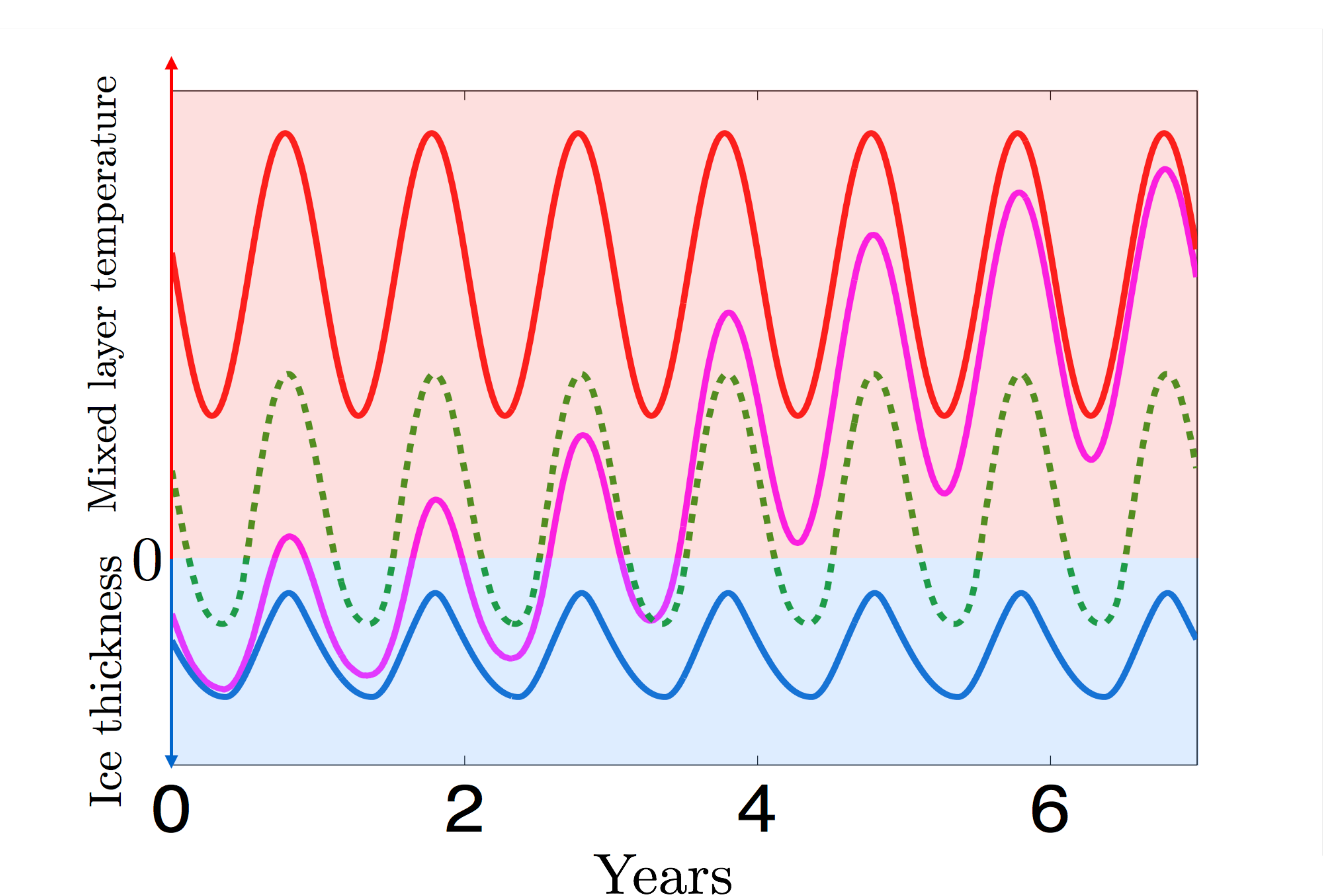}
\caption{}
\end{subfigure}
\begin{subfigure}[b]{0.16\textheight}
\includegraphics[width=\textwidth]{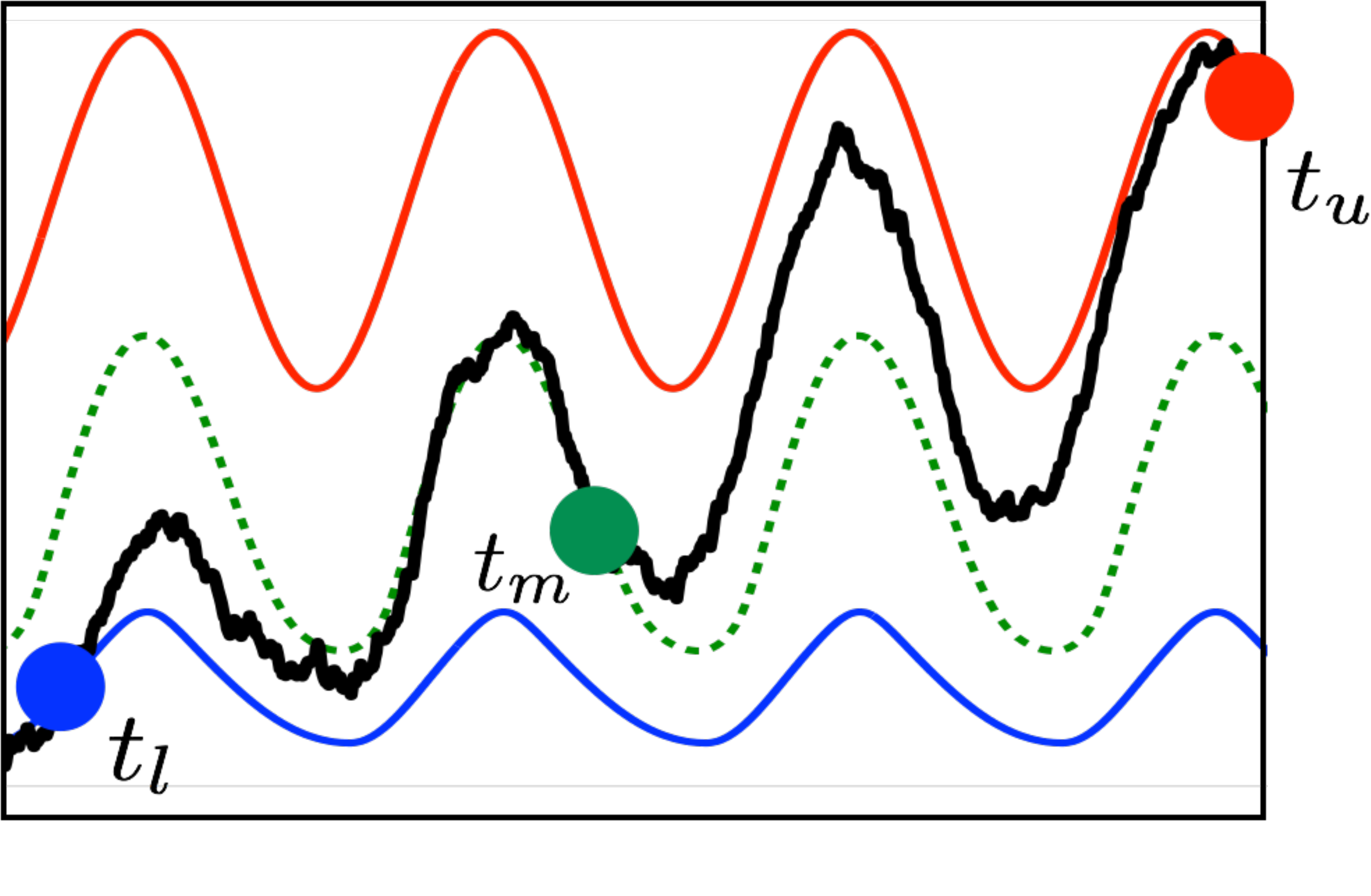}
\caption{}
\end{subfigure}
\begin{subfigure}[b]{0.16\textheight}
\includegraphics[width=\textwidth]{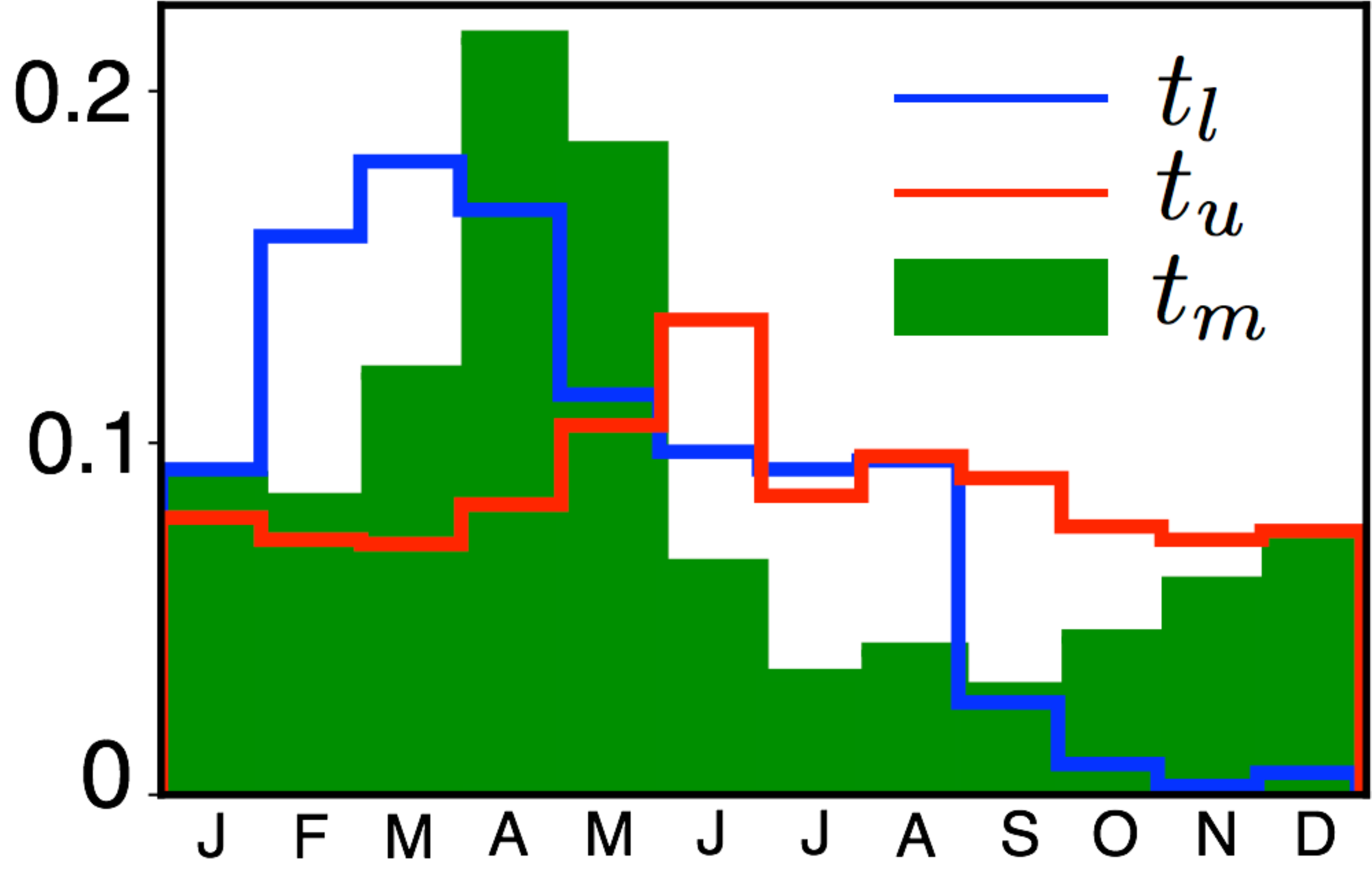}
\caption{}
\end{subfigure}
\caption{Summary of noise-induced tipping events for a version of the Arctic sea ice model \cite{eisenman2012factors}. (a): Vertical scale is surface enthalpy  density, proportional to ocean mixed-layer temperature when positive (red) and ice thickness (blue) when negative. Red (blue) curves are the stable perennially ice-free (ice-covered) state; green curve is the unstable seasonally ice-free state. Magenta curve is the mean of $832$ transition samples. (b-c): Percentage histograms of the departure times from neighborhoods of the perennially ice-covered state ($t_l$) in blue and the seasonal state ($t_m$) in green, and the arrival times at a neighborhood of the perennially ice-free state ($t_u$) in red;  
neighborhood size chosen so that the distribution is invariant when the size is halved. Model parameters are taken from Table 1 in the paper by Eisenman \cite{eisenman2012factors} with $\tilde{L}_m=1.1$, and $\sigma=0.17$. Each simulation runs for $1000$ years. ($E= 0.4\%$ in the Monte Carlo convergence test.) }
\label{fig:seaIce}
\end{figure}
Our preliminary investigations were based on an Arctic sea ice model due to Eisenman and Wettlaufer  \cite{eisenman2009nonlinear}. Their model, an energy-balance column model of the Arctic,  captures the strong seasonal variation in incoming solar and outgoing long-wave radiation over the course of the year, and the key feedbacks associated with the ice-albedo and the sea ice thermodynamics. It is a 1-D nonautonomous dynamical system which shows bistability between the current  Arctic state and one where the Arctic is ice-free. 
We illustrate the noise-induced tipping phenomenon in this setting, using a version of the model due to Eisenman \cite{eisenman2012factors}, with additive white noise. A stochastic version of this model has also been investigated by Moon and Wettlauffer \cite{moon2017stochastic}, who consider both additive and multiplicative noise. 

The Eisenman-Wettlaufer (EW) model \cite{eisenman2009nonlinear} has a prominent hysteresis loop under changes in greenhouse gases, akin to Figure \ref{Fig:DynamicsDiagram}(a). Here we consider the parameters in this model corresponding to the current state of the Arctic, which has perennial ice-cover and is bistable, in this model, with a perennially ice-free state; see the deterministic orbits in Figure \ref{fig:seaIce}(a) and (b). Figure \ref{fig:seaIce}(c) presents distributions obtained from Monte Carlo simulations of the model with additive white noise. In contrast with our findings of the histograms for \eqref{eqn:ModelSDE}, summarized by Figure \ref{fig: histogramsVarySigma}, which showed the most prominent peak in the distribution for $t_l$,  we find that the distribution of $t_m$ is most peaked for the EW model, with a peak in April, which is, interestingly, well before the summer melting season. In Figure \ref{fig:seaIce}(c) we shift all the tipping samples to cross the seasonal state in the third year and plot their mean (magenta curve). We note  that, on average, the tipping events take many years to complete the transition, even with relatively strong noise intensity. Moreover the noisy trajectories often stay in a neighborhood of the unstable seasonally ice-free state for almost a year or more before leaving for the perennially ice-free state, making it challenging to identify `a tipping phase', despite the peaked distributions. These simulations appear to be in 
a parameter regime where the deterministic relaxation dynamics are weak compared to the noise strength driving the tipping events. We also note that the model has a natural non-smooth limit, leading to a discontinuity boundary at $E = 0$,  due to the transition from ice-covered to ice-free dynamics \cite{hill2016analysis}. To further investigate its tipping events,  considering the model switches that take place, with the seasonal effects, along the unstable seasonally ice-free state may be important.}

\section*{Acknowledgements}
We thank Julie Leifeld for some early discussions. AV thanks Bj\"{o}rn Sandstede for his guidance in AUTO. We thank Paul Ritchie and Jan Sieber for introducing us to the path integral approach. MS, YC, and AV are partially funded by the AMS Mathematics Research Communities (NSF Grant No. 1321794). 
The work of MS and YC is funded in part by NSF DMS-1517416. The work of AV and JG was supported in part by NSF-RTG grant DMS-1148284, and AV is currently supported by the Mathematical Biosciences Institute and the NSF under DMS- 1440386.

\bibliographystyle{plain}
\bibliography{references}


\appendix

\section{Convergence of OM functional to FW functional in the limit $\sigma\rightarrow 0$} \label{sec:TheoremApend}
\begin{theorem} \label{Thm:ConvergenceFM} For $\sigma \in [0,1]$, if $x_{\sigma}\in \mathcal{A}$ is a sequence satisfying $\min_{x \in \mathcal{A}} I_{\sigma}[x] =I_{\sigma}[x_{\sigma}]$ then there exists a subsequence $x_{\sigma_k}\in \mathcal{A}$ with $\sigma_k\rightarrow 0$ as $k\rightarrow \infty$ such that:
\begin{equation*}
\lim_{k \rightarrow \infty}I_{\sigma_k}[x_{\sigma_k}]=\min_{x\in \mathcal{A}}I_\text{FW}[x].
\end{equation*}
Moreover, there exists $x_0\in \mathcal{A}$ such that $x_{\sigma_k}\stackrel{H^1}{\rightarrow} x_0$ and $\min_{x\in \mathcal{A}} I_\text{FW}[x]=I_\text{FW}[x_0]$.
\end{theorem}
\begin{proof}
By Corollary \ref{Cor:Existence} there exists $x_{\sigma}\in \mathcal{A}$ that minimizes $I_{\sigma}$ and thus for all $x \in \mathcal{A}$
\begin{equation}
\limsup_{\sigma \rightarrow 0} I_{\sigma}[x_{\sigma}]\leq \limsup_{\sigma \rightarrow 0}I_{\sigma} [x]=I_\text{FW}[x]. \label{Eq:Proof1}
\end{equation}
Therefore, Theorem \ref{Thm:Coercivity} implies that $\|\dot{x}_{\sigma}\|_{L^2}$ is bounded and by the Poincar\'e inequality $x_{\sigma}$ is bounded with respect to the $H^1$ norm. Consequently, by the Banach-Alaoglu theorem there exists $x_0\in \mathcal{A}$ and a subsequence $x_{\sigma_k}$ with $\sigma_k\rightarrow 0$ as $k\rightarrow \infty$ such that $x_{\sigma_k}\stackrel{H^1}{\rightharpoonup} x_0$. Moreover,
\begin{align*}
&\frac{\sigma_k^2}{\varepsilon}\left[(t_f-t_0)-3\sup_{k}\|x_{\sigma_k}\|_{L^2}^2\right]\\\leq &\frac{\sigma_k^2}{\varepsilon}I_{\text{OM2}}[x_{\sigma_k}]\\=&\frac{\sigma_k^2}{\varepsilon}\int_{t_0}^{t_f}\left(1-3x_{\sigma_k}(t)^2\right)dt\leq \frac{\sigma_k^2}{\varepsilon}(t_f-t_0)
\end{align*}
and therefore 
\begin{equation}
\lim_{k\rightarrow \infty}\frac{\sigma_k^2}{\varepsilon}I_{\text{OM2}}[x_{\sigma_k}]=0. \label{Eq:Proof2}
\end{equation}

By Corollary \ref{Cor:Existence} there exists $\bar{x}_0\in \mathcal{A}$ that minimizes $I_{\text{FW}}$. Since $I_{\text{FW}}$ is convex in $\dot{x}$ it is weakly lower semi-continuous with respect to the $H^1$ inner product \cite{jost1998calculus}. Therefore, by lower semi-continuity, (\ref{Eq:Proof1}), and (\ref{Eq:Proof2}) it follows that
\begin{align*}
I_\text{FW}[x_0]&\leq \liminf_{k \rightarrow \infty} I_{\text{FW}}[x_{\sigma_k}] = \liminf_{k \rightarrow \infty}I_{\sigma_k}[x_{\sigma_k}]\\
&\leq  \limsup_{k\rightarrow \infty} I_{\sigma_k}[x_{\sigma_k}] \leq I_\text{FW}[\bar{x}_0].
\end{align*}
Since $\bar{x}_0$ minimizes $I_\text{FW}$ all of the above inequalities are in fact equalities and therefore $x_0$ is a minimizer of $I_\text{FW}$. Furthermore,  it follows that $\liminf_{k \rightarrow \infty}I_{\sigma_k}[x_{\sigma_k}]=\limsup_{k\rightarrow \infty} I_{\sigma_k}[x_{\sigma_k}]$, 
which implies $\lim_{k\rightarrow \infty}I_{\sigma_k}[x_{\sigma_k}]$ exists and satisfies $\lim_{k\rightarrow \infty}I_{\sigma_k}[x_{\sigma_k}]=I_{\text{FW}}[x_0]=I_{\text{FW}}[\bar{x}_0]$.

Finally, since $x_{\sigma_k}\in H^1([t_0,t_f];\mathbb{R})$ it follows that $x_{\sigma_k}$ is absolutely continuous and hence $|x_{\sigma_k}(t)-x_0(t)|\leq \int_{t_0}^{t_f}\left|\dot{x}_{\sigma_k}(t)-\dot{x}_0(t)\right|dt$.
Therefore, it follows from weak convergence of $\dot{x}_{\sigma_k}$ that $x_{\sigma_k}$ converges uniformly to $x_{0}$. Expanding it follows that
\begin{align*}
&\left|\|\dot{x}_{\sigma_k}\|_{L^2}^2-\|\dot{x}_{\sigma}\|_{L^2}^2\right|\\
= &\left|I_{\sigma_k}[x_{\sigma_k}]-I_{\text{FW}}[x_{\sigma_k}]-2\int_{t_0}^{t_f}\dot{x}_{\sigma_k}(t)f(x_{\sigma_k}(t),t)dt\right.\\
&\left.\,\, +\int_{t_0}^{t_f}f(x_{\sigma_k}(t),t)^2dt+2\int_{t_0}^{t_f}\dot{x}_0(t)f(x_0(t),t)dt\right.\\
&\left.\,\,-\int_{t_0}^{t_f}f(x_0(t),t)^2dt-\frac{\sigma^2}{\varepsilon}\int_{t_0}^{t_f}f_x(x_{\sigma_k}(t),t)dt\right|.
\end{align*}
Therefore, from uniform convergence of $x_{\sigma_k}$, weak convergence of $\dot{x}_{\sigma_k}$ and (\ref{Eq:Proof2}) it follows that $\lim_{k\rightarrow \infty} \|\dot{x}_{\sigma_k}\|_{L^2}=\|\dot{x}_{0}\|_{L^2}$ and thus we can conclude that $\dot{x}_{\sigma_k}\rightarrow \dot{x}_{0}$ strongly in $L^2$ \cite{evans1990weak}. Hence,  by the Poincar\'e inequality $x_{\sigma_k}$ converges strongly to $x_{0}$ in $H^1$ .
\end{proof}

\section{Convergence test of Monte Carlo simulations and more comparison results}\label{sec:MCappendix}
\begin{figure}[ht]
\centering
\includegraphics[width=0.46\textwidth]{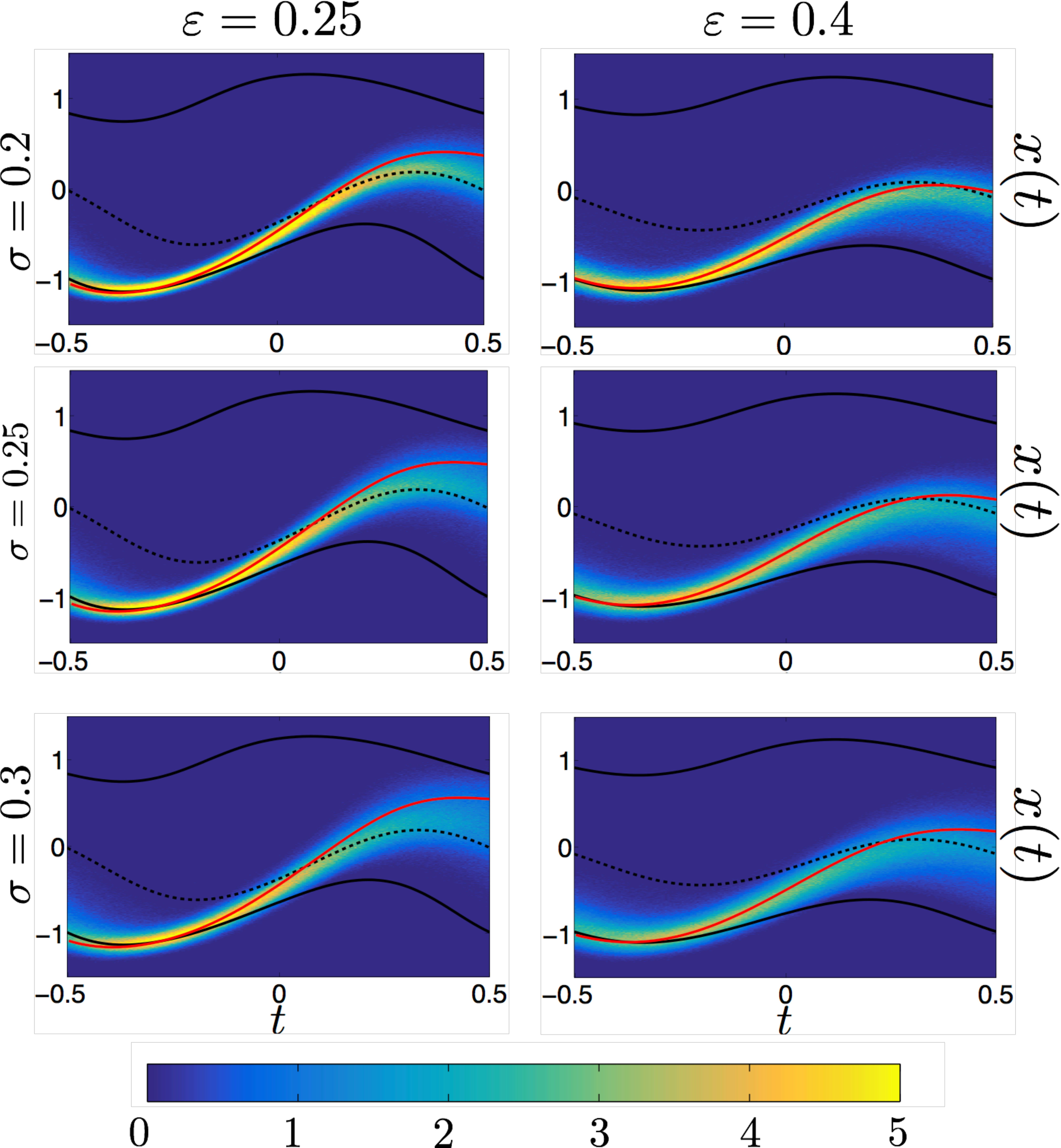}
\caption{Comparison between the optimal path results, represented by the red curves, and the Monte Carlo results. Parameters are taken as follows: $\varepsilon = 0.25$ on the left column and $\varepsilon = 0.4$ on the right one; $\sigma = 0.2, 0.25$ and $0.3$ from top to the bottom rows. $\alpha$ and $A$ are fixed at $0.15$ and $0.7$, respectively. These are illustrations of the result summarized in Table \ref{table: MCcomparison}. Our criterion for convergence test of the Monte Carlo simulation is $E<10^{-5}$, where $E$ is defined in \eqref{eqn: meansquarederror}. From top to bottom, $E = 1.2\times 10^{-6}, 1.2\times 10^{-6}, 7.5\times 10^{-6}$ on the left column and $E = 8.7\times 10^{-6}, 1.4\times 10^{-6}, 1.7\times 10^{-6}$ on the right column. }\label{fig: more_MC_comparison}
\end{figure}
Here we describe the procedure of Monte Carlo simulations. We use the Euler-Maruyama method \cite{higham2001algorithmic} on the interval $[t_0,t_f]$ with initial data $x_0=x_l^*(t_0)$ to generate a sequence of samples ${X_j(t)}$ for (\ref{eqn:ModelSDE}), each of which runs for $K$ time steps. We collect N samples that meet the following criteria for transitioning from $x_l^*$ to $x_u^*$. Let  $\tau_j=\inf_{t}\{X_j(t)> x_u^*(t)\}$ denote the first time a sample path crosses $x_u^*(t)$. We define \emph{the tipping events} to be $\{X_i\}$ that satisfy $\{X_i\}=\{X_j: \tau_j<t_f\}$. We then construct the distribution of the tipping events using $M$ bins. The convergence test of the Monte Carlo simulations used to obtain distributions of the tipping events shown in Section \ref{sec:noiseDriftBalance} involves the following steps:
\begin{enumerate}
\item At each time step, $t^{(k)}$ ($k = 1...K$), we divide the samples $\{X_i(t^{(k)}): i = 1... N\}$ into $M$ equal width bins in the $x$ direction. We then record  the fraction $h_N^{(m,k)}$ of the samples that fall into the $m^\text{th}$ bin. We pick $M$ in the histogram based on Freedman Diaconis rule \cite{freedman1981histogram}, and $K$ so that Euler-Maruyama method converges.
\item We double the number of samples to $2N$ and repeat Step 1 to get $h_{2N}^{(m,k)}$.
\item We then compute the mean squared error between $h_{2N}^{(m,k)}$ and $h_{N}^{(m,k)}$ ($m = 1...M, k = 1...K$), given by $E = \sum_{k = 1}^{K}\sum_{m = 1}^M \left(h_{N}^{(m,k)} - h_{2N}^{(m,k)}\right)^2/(MK)$. 
\item If $E<10^{-5}$, we say that the Monte Carlo simulation has converged. If it has not converged, we double the number of sample paths and repeat the procedure.
\end{enumerate}

\section{Independence of initial and final times $t_0$ and $t_f$}\label{sec:independenc}
\begin{figure}[ht]
\begin{center}
\includegraphics[width=0.47\textwidth]{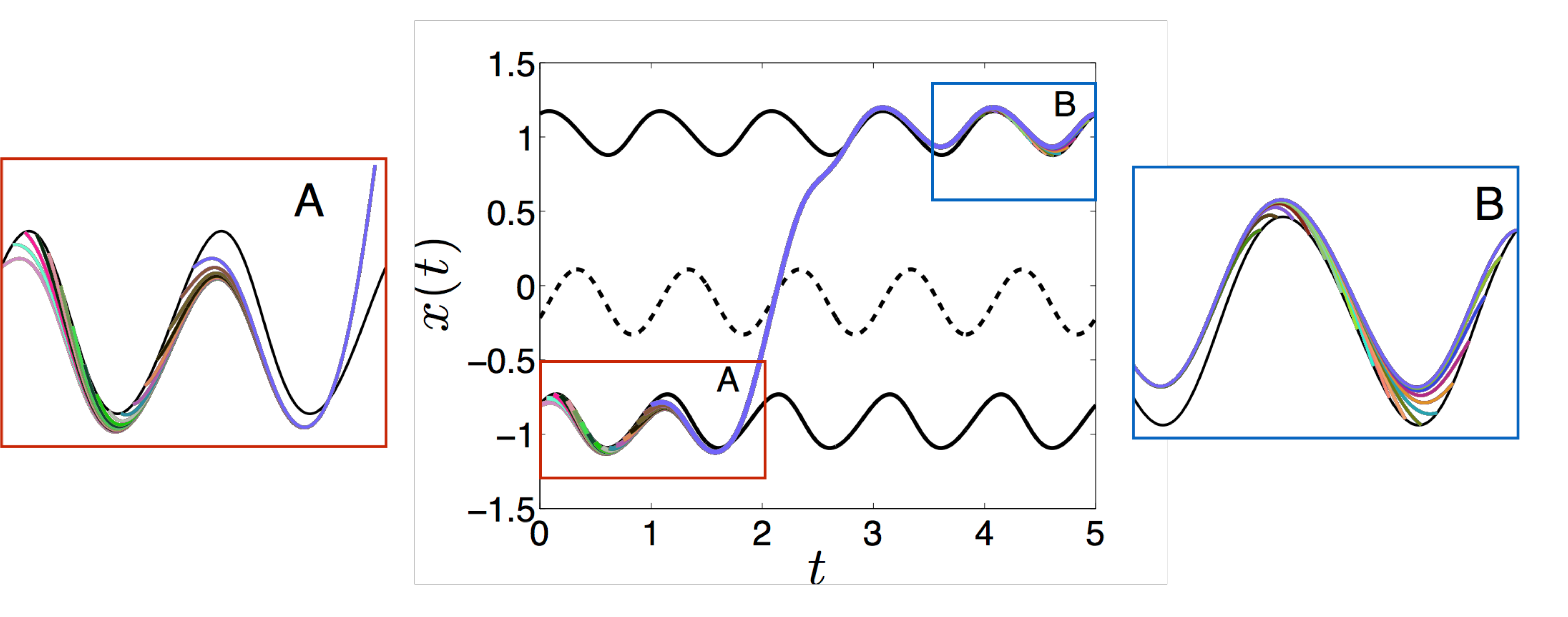}
\end{center}
\caption{Illustration of 
numerical experiment to test 
that the optimal path is independent of $t_0$ and $t_f$.
We vary $t_0$ and $t_f$ in the first and third periods. 
$(\varepsilon,\alpha,A,\sigma)= (0.25, 0.1, 0.4,0.5)$.}\label{fig:independence}
\end{figure}
To show that the optimal path is not sensitive to the choice of  $t_0$ and $t_f$, we vary $t_0$ ($t_f$) by taking $16$ uniformly spaced values in $[0, 1]$ ($[4, 5]$).
Initial guesses are in the form of \eqref{eq:initialGuess} with the same intermediate time $t_j = 2$. In blow up box A (B), curves with different colors indicate that the paths start (end) at different positions on $x_l^*$ ($x_u^*$).  
In the second period all paths converge to the same 
curve from $x_l^*$ to $x_u^*$. After arriving at $x_u^*$, 
the paths start to separate shortly before the fourth period since we impose different final positions.
Thus, given a large enough interval of time between $t_l$ and $t_u$, optimal paths which start and end at different positions transition from $x_l^*$ to $x_m^*$ at the same forcing phase.

\end{document}